\documentclass[11pt]{article} 

\usepackage[utf8]{inputenc} 

\usepackage{geometry} 
\usepackage{graphicx, subfig} 

\usepackage{booktabs} 
\usepackage{array} 
\usepackage{paralist} 
\usepackage{verbatim} 
\usepackage{subfig} 

\usepackage{fancyhdr} 
\usepackage{sectsty}
\allsectionsfont{\sffamily\mdseries\upshape} 

\usepackage[nottoc,notlof,notlot]{tocbibind} 
\usepackage[titles,subfigure]{tocloft} 

\usepackage{cite}

\usepackage{amsmath, amsfonts, amsthm, amssymb, amscd}
\usepackage{indentfirst}

\usepackage[pdftex,bookmarksnumbered]{hyperref}

\newtheorem{theorem}{Theorem}

\newtheorem{proposition}{Proposition}

\title{A Note on Hardness Frameworks and Computational Complexity of Xiangqi and Janggi}
\author{Zhujun Zhang \thanks{E-mail: zhangzhujun1988@163.com} \\Water Bureau of Fengxian District, Shanghai}

\date{March 30, 2019} 

\providecommand{\keywords}[1]{\textbf{\textit{Index terms---}} #1}

\begin{document}
\maketitle

\begin{abstract}
We review NP-hardness framework and PSPACE-hardness framework for a type of 2D platform games. 
We introduce a EXPTIME-hardness framework by defining some new gadgets. 
We use these hardness frameworks to analyse computational complexity of Xiangqi (Chinese Chess) and Janggi (Korean Chess). 
We construct all gadgets of the hardness frameworks in Xiangqi and Janggi.
In conclusion, we prove that Xiangqi and Janggi are both EXPTIME-complete.

\end{abstract}


\section{Introduction}

In last years, researchers introduce some useful hardness framework for proving hardness of 2D platform games.
With these frameworks in hand, we can analyse computational complexity of games conveniently.
To prove hardness of games, we just need to construct all gadgets of hardness frameworks in games.
Recently, complexity of many 2D platform games have been determined by using these frameworks.

Xiangqi (Chinese Chess) is a strategy board game for two players, and it is recognized as a classic Chess variant.
Xiangqi is one of the most popular board games in China. 
Janggi (Korean Chess) is the variant of Chess played in Korea, and it derived from Xiangqi.
Janggi is highly similar to Xiangqi, but it still differs from Xiangqi in various ways.
There is little research on complexity of Xiangqi and Janggi.

In this note, we review previous work in Section 2;
in Section 3, we describe three hardness frameworks;
we discuss complexity of Janggi and Xiangqi in Section 4 and Section 5 respectively; 
we summarize this note in Section 6.

\newpage

\section{Previous Work}

In last decades, many classic two-player board games were proved to be computationally hard.
In 1980, Reisch \cite{gomoku} proved Gomoku (Gobang) to be PSPACE-complete.
In 1984, Robson \cite{checkers} proved Checkers to be EXPTIME-complete.
In 1994, Iwata and Kasai \cite{othello} proved Othello (Reversi) to be PSPACE-complete.

Computational complexity of Go and subproblems of the game were analysed in depth.
In 1980, Lichtenstein and Sipser \cite{gopspace} proved Go with the superko rule to be PSPACE-hard. 
In 1983, Robson \cite{goexptime} proved  Go with the basic ko rule to be EXPTIME-complete.
In 1998, Cr{\^a}{\c{s}}maru \cite{tsumeGo} proved Tsume-Go to be NP-complete.
In 2000, Cr{\^a}{\c{s}}maru and Tromp \cite{goladders} proved Ladders to be PSPACE-complete.
In 2002, Wolfe \cite{goendgames} proved Go endgames to be PSPACE-hard.

People also researched computational complexity of Chess and variants of Chess. 
In 1981, Fraenkel and Lichtenstein \cite{chess} proved Chess to be EXPTIME-complete.
In 1987, Adachi et al. \cite{shogi} proved Shogi to be EXPTIME-complete.
In 2001, Yokota et al. \cite{tsumeshogi} proved Tsume-Shogi to be EXPTIME-complete.
In 2005, Yato et al. \cite{yozumetsumeshogi} proved finding another solution of Tsume-Shogi to be EXPTIME-complete.

Demaine and Hearn review results
about the complexity of many classic games in their survey paper \cite{playinggameswithalgorithms}.
Other results about complexity of well-known games could be found in Wiki page ``Game complexity'' \cite{gamecomplexity}.

More recently, researchers focus on complexity of some interesting 2D platform video games and frameworks which could be used to prove hardness of games.
In 2010, Fori\v{s}ek \cite{2dplatformgames} introduced a basic NP-hardness framework for 2D platform games.
In 2012, Aloupis, Demaine and Guo \cite{nintendogamesnphard} introduced a elegant NP-hardness framework, and they used the framework to analyse complexity of several classic Nintendo games.
In 2014, Viglietta \cite{gamehardjob} established some general schemes relating the computational complexity of 2D platform games. 
Viglietta applied such schemes to several games, for instance, he proved Lemmings to be PSPACE-complete in \cite{lemmings}.
In 2015, Aloupis et al. \cite{nintendogameshard} introduced a new PSPACE-hardness framework for games.
In this framework, usually, only open-close door and crossover gadgets are nontrivial.
Using such framework, Demaine, Viglietta and Williams \cite{supermariohard} proved Super Mario Bros. to be PSPACE-complete in 2016.
Other results about complexity of games and puzzles could be found in Hearn and Demaine's summary book \cite{gamepuzzlecomputation}.

There is little research on complexity of Xiangqi and Janggi.
Park \cite{statecomplexityxj} discussed space-state complexity of Xiangqi and Janggi.
Khan \cite{analysischinesechess} used combinatorial game theory (CGT) to analyse complexity of Xiangqi.
Gao and Xu \cite{gaophd} \cite{chinesechess} argued that Xiangqi is EXPTIME-complete, however, there is a mistake in the Boolean Controller of their reduction.
Our method in this note is fully inspired by their work.

\newpage

\section{Hardness Frameworks}

We assume that each player control a moveable avatar in a game.
The general decision problem is to decide whether the avatar gets from a given start point to a given finish point (or before its opponent does).
In this section, we describe three hardness frameworks for proving hardness of 2D platform games.  
With these frameworks in hand, we can easily establish reduction by constructing necessary gadgets in games.

\subsection{NP-hardness Framework}

We describe the first framework for proving NP-hardness of games. 
The framework is from \cite{nintendogamesnphard} and \cite{nintendogameshard}. 
The framework reduces from the classic NP-complete problem 3-SAT.
To prove NP-hardness, we need following gadgets:

\textbf{Start gadget.} 
The start gadget contains the start point and one exit for the avatar. 

\textbf{Finish gadget.} 
The finish gadget contains one entrance and the finish point for the avatar. 
When the avatar enters the finish gadget, the player wins immediately.

\textbf{Turn gadget.} 
The turn gadget contains one entrance and one exit for the avatar.
The avatar can change moving direction by traversing a turn gadget.

\textbf{Switch and merge gadgets.} 
The switch gadget contains one entrance and more than one exits for the avatar.
When the avatar enters a switch gadget, it can choose one of exits to leave the gadget.
On the other hand, the merge gadget contains one exit and more than one entrances for avatar.
The avatar enters a merge gadget from the entrances, then it leaves the gadget through the only exit.
In many games, merge gadgets are identical to switch gadgets .

\textbf{One-way gadget.} 
The one-way gadget contains one entrance and one exit for the avatar.
The avatar can only traverse this gadget in one direction, and traversing in reverse direction is forbidden.
We use lines with arrows to represent one-way gadgets in figures.

\textbf{Crossover gadget.} 
The crossover gadget contains two paths that cross each other.
The avatar can traverse each path respectively, and there is no leakage between two paths.

\textbf{Door gadget.} 
A door gadget is an object with two states, open and closed.
The door gadget contains two paths for the avatar, and there is no leakage between two paths.
One path is an open path, and the other is a traverse path.
The avatar can traverse the traverse path if and only if the door gadget is in the open state.
When the avatar traverses the open path, it is allowed to open the door gadget.

\begin{figure}[t]
	\centering  
	\includegraphics[width=0.7 \linewidth]{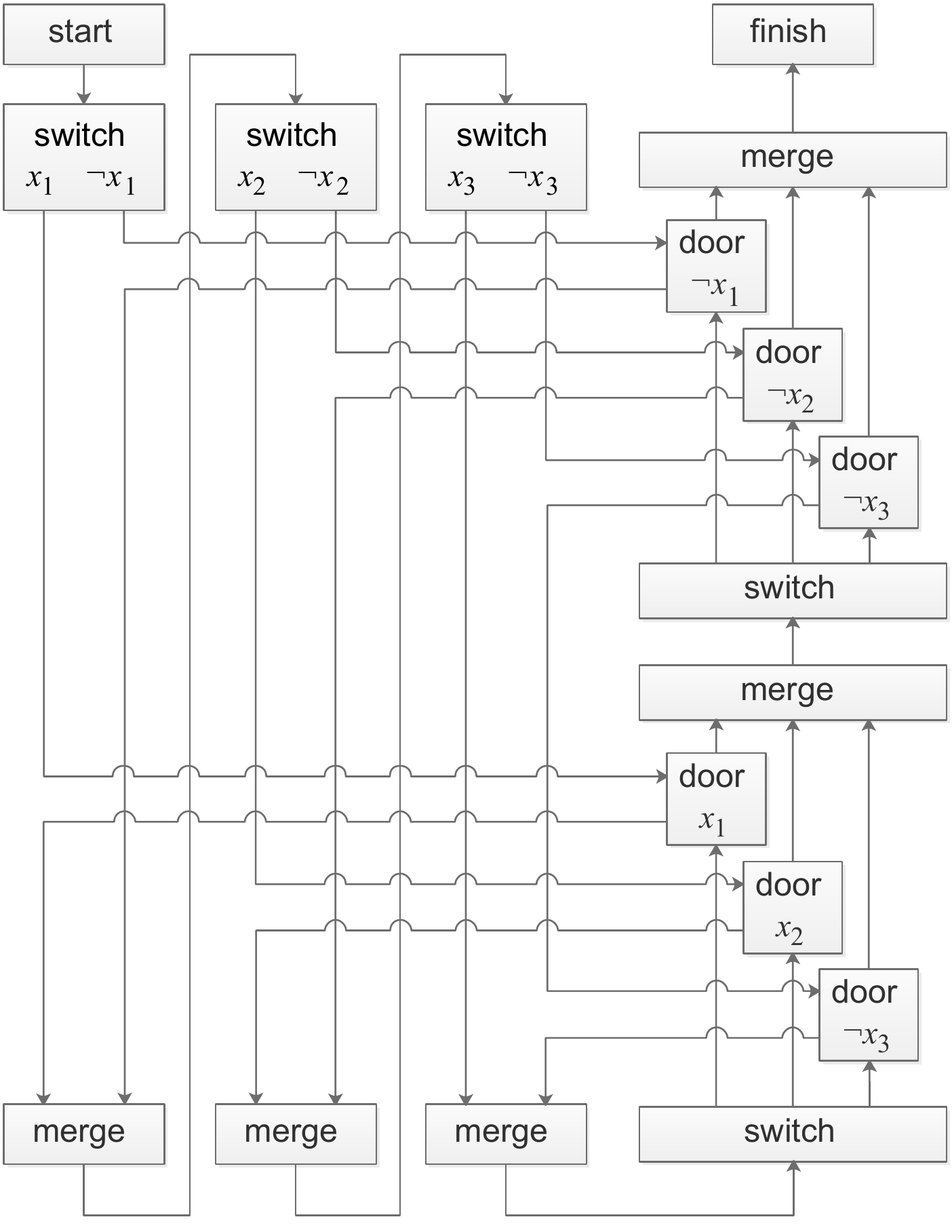}  
	\caption{NP-hardness framework.}  
	\label{NPHframework}   
\end{figure}

We use a example to describe the hardness framework.
Figure \ref{NPHframework} illustrates the NP-hardness framework corresponding to formula $\left( \neg x_1 \vee \neg x_2 \vee \neg x_3 \right) \wedge \left( x_1 \vee x_2 \vee \neg x_3 \right) $.
At the beginning, all door gadgets are in the closed state.
The player-controlled avatar starts from the start gadget.
Then it enters the first switch gadget corresponding to the variable $x_1$ in the formula.
The avatar may choose one exit to leave the switch gadget, which is corresponding to assignment of $x_1$.
After the avatar leaves the switch gadget, the one-way gadget prevents it from moving back to the switch gadget.
Then the avatar can open the door gadgets corresponding to the literals in the formula.
After opening the door gadgets, the avatar enters a merge gadget, which indicates variable $x_1$ has been assigned. 
Then the avatar enters the next switch gadget corresponding to variable  $x_2$.
This process continues until all variables have been assigned.
The avatar enters a switch gadget with three exits, then it has to choose one exit connecting to an open door gadget.
Traversing a open door gadget indicates that the corresponding clause in the formula is true.
The situation of another clauses is similar.
Finally, the avatar arrives the finish gadget.

It is easy to verify that the 3-CNF boolean formula can be made true if and only if the avatar can arrive the finish gadget in a game.
Moreover, number of gadgets in the framework is polynomial, so that the reduction could be established in polynomial time.

\textbf{Remarks.}
Since every path of NP-hardness Framework will be traversed by the avatar for at most one time. 
The reduction still works even if all gadgets are single-use gadgets.
Usually, it is easy to construct the start, finish, turn, switch and merge gadgets in games, while constructing the crossover, one-way and door gadgets is more difficult.

\subsection{PSPACE-hardness Framework}

We describe the second framework for proving PSPACE-hardness of games. 
The framework is from \cite{nintendogameshard} and \cite{supermariohard}. 
The framework reduces from the PSPACE-complete problem True Quantified Boolean Formula (TQBF).
We need the start, finish, turn, switch, merge, one-way, crossover and open-close door gadgets to establish PSPACE-hardness framework.
Most gadgets of the framework are described in last section.
Thus we just describe the open-close door gadget and a auxiliary gadget here.

\textbf{Open-close door gadget.} 
Figure \ref{openclosedoor} illustrates an open-close door gadget of PSPACE-hardness framework. 
An open-close door gadget is an object with two states, open and closed.
The open-close door gadget contains three paths, open path, traverse path and close path.
The avatar can traverse the traverse path if and only if the open-close gadget is in the open state.
When the avatar traverses the open path, it is allowed to open the gadget; when the avatar traverses the close path, it has to close the gadget.
Moreover, there is no leakage between three paths.

\begin{figure}[htbp]
	\centering  
	\includegraphics[width=0.4 \linewidth]{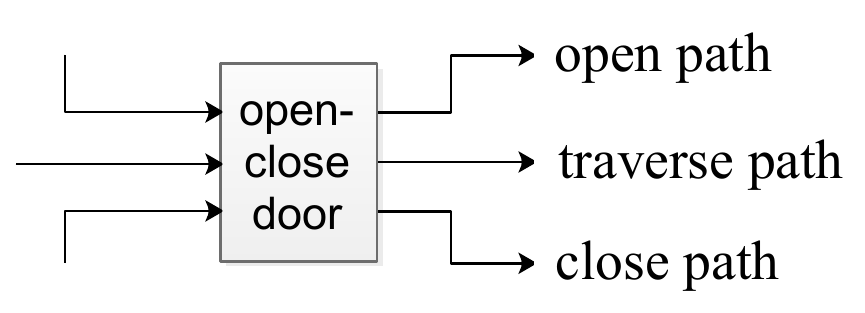}  
	\caption{Open-close door gadget.}  
	\label{openclosedoor}   
\end{figure}

\begin{figure}[htbp]
	\centering  
	\includegraphics[width=0.35 \linewidth]{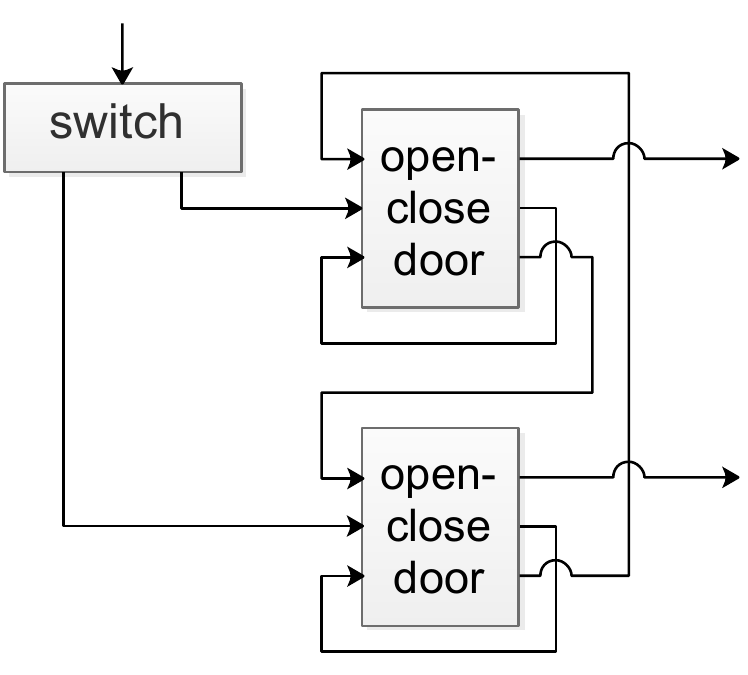}  
	\caption{Alternation gadget.}  
	\label{alternation}   
\end{figure}

\begin{figure}[htbp]
	\centering  
	\includegraphics[width=1.0 \linewidth]{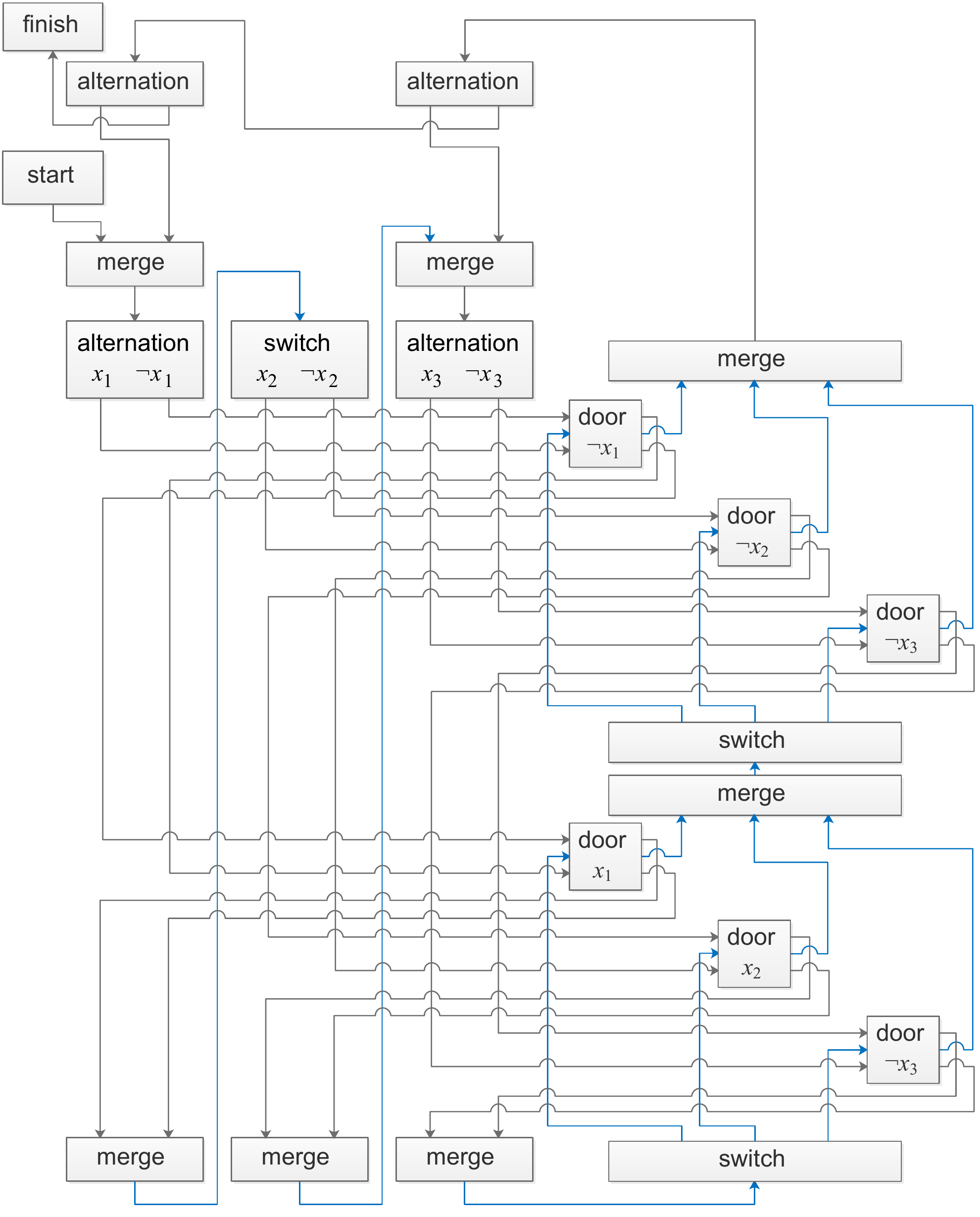}  
	\caption{PSPACE-hardness framework.}  
	\label{PSPACEHframework}   
\end{figure}

\textbf{Alternation gadget.} 
The alternation gadget contains one entrance and two exits for the avatar.
The avatar leaves the alternation gadget through these two exits alternately.
When the avatar enters the gadget for the first time, it can only leave the gadget through one of the exits; when the avatar enters the gadget for the second time, it can only leave the gadget through another exit, and so forth.
The alternation gadget could be constructed by other gadgets.
Figure \ref{alternation} illustrates an alternation gadget which is constructed by the turn, switch, crossover and open-close door gadgets.
Every time the avatar traverse the alternation gadget, it has to close a open-close door gadget, and then it opens another open-close door gadget.

Again, we use a example to describe the hardness framework. 
Figure \ref{PSPACEHframework} illustrates the framework corresponding to formula $ \forall x_1 \exists x_2 \forall x_3 \left(  \left( \neg x_1 \vee \neg x_2 \vee \neg x_3 \right) \wedge \left( x_1 \vee x_2 \vee \neg x_3 \right) \right) $.
The player-controlled avatar starts from the start gadget.
Then it enters the alternation gadget corresponding to the universal variable $x_1$ in the formula.
As it is the first time that avatar enters the gadget, the avatar has to leave gadget through the left exit, which indicates that $x_1$ is assigned true.
Then the avatar opens the open-close door gadgets corresponding to literal $x_1$, on the other hand, it closes the open-close door gadgets corresponding to literal $ \neg x_1$.
After operating door gadgets, the avatar enters a merge gadget and a switch gadget corresponding to the existential variable $x_2$.
This process continues until all variables has been assigned.
Then the avatar enters the check progress.
For each clause in the formula, the avatar has to choose one of three paths (coloured in blue) to traverse.
After checking all clauses, the avatar enters a alternation gadget.
As it is the first time that avatar enters this gadget, the avatar has to leave gadget through the left exit, which forces the avatar to enters the alternation gadget corresponding to $x_3$ again.
Since it is the second time that the avatar entering the gadget, the avatar has to leave gadget through the right exit, which indicates that $x_3$ is assigned false.
Then the avatar enters the check progress again.
This process simulates the universal variable in the formula.
Finally, the avatar arrives the finish gadget if all possible assignments have been tested.

It is also easy to verify that the formula is quantified if and only if the avatar can arrive the finish gadget.
Moreover, number of gadgets in the framework is polynomial, so that the reduction could be established in polynomial time.

\textbf{Remarks.}
The gadgets of PSPACE-hardness framework will be traversed for exponential times potentially, thus it is insufficient to prove PSPACE-hardness by constructing single-use gadgets.
Sometimes, an open-close door gadget may contain additional paths connecting to the finish gadget, and we will see this type of gadgets in Section 4 and 5.

\subsection{EXPTIME-hardness Framework}

Using similar technique, we introduce a EXPTIME-hardness framework for 2D platform games here.
The framework reduces from EXPTIME-complete formula game $G_2$ described in \cite{difficultgames}.

$G_2$ is a two-player formula game where players move alternately.
A position of $G_2$ is a 4-tuple $(\tau, \text{I-WIN} (X,Y), \text{II-WIN} (X,Y), \alpha )$ where $\tau \in \{1,2\}$, $X = \{x_1, x_2, ...\}$, $Y = \{y_1, y_2, ...\}$, I-WIN and II-WIN are formulas in 12DNF, and $\alpha$ is an $(X \cup Y)$-assignment.
Player I(II) moves by changing the value assigned to at most one variable in $X(Y)$; either player may pass since changing no variable amounts to a ``pass''.
Player I(II) wins if the formula I-WIN(II-WIN) is true after some move of player I(II).
More precisely, player I can move from $(1, \text{I-WIN} (X,Y), \text{II-WIN} (X,Y), \alpha )$  to $(2, \text{I-WIN} (X,Y), \text{II-WIN} (X,Y), \alpha ')$ in one move if and only if $\alpha '$ differs from $\alpha$ in the assignment given to at most one variable in $X$ and II-WIN is false under the assignment $\alpha$; the move of player II are defined symmetrically.
Deciding whether Player I has a forced win in $G_2$ is EXPTIME-complete.

To simulate $G_2$, we assume that there are two players in a 2D platform game, and we use Red and Black to represent two players.
Red and Black control the red avatar and the black avatar respectively, and they move alternately.
To establish EXPTIME-hardness framework, we need the start, finish, turn, switch, merge, one-way gadgets for Red and Black respectively, and we need three types of crossover gadgets and six types of door gadgets.
The start, finish, turn, switch, merge, one-way gadgets for Red are identical to the gadgets described in Section 3.1, and the gadgets for Black could be constructed symmetrically in many games.
The crossover gadget for Red and Red is identical to the gadget described in Section 3.1, and crossover gadget for Black and Black could be constructed symmetrically.
The crossover gadget for Red and Black contains one path for Red and one path for Black, the red avatar and the black avatar can respectively traverse the gadget without obstructing each other.
Six door gadgets are RRR door, BBB door, RBR door, BRB door, BBR door and RRB door gadgets.
We use letters ``R'' and ``B'' to represent paths for red avatar and black avatar respectively.
The RRR door gadget is identical to the open-close door gadget described in last section.
Here, we just describe the RBR door and BBR door gadgets, since the BBB door, BRB door and RRB door gadgets could be defined symmetrically.

\textbf{RBR door gadget.} 
A RBR door gadget is an object with two states, open and closed.
The RBR door gadget contains three paths, open path, traverse path and close path.
The black avatar can traverse the traverse path if and only if the RBR gadget is in the open state.
When the red avatar traverses the open path, it has to open the gadget.
When the avatar traverses the close path, it is allowed to close the gadget.
Moreover, there is no leakage between three paths.

\begin{figure}[htbp]
	\centering  
	\includegraphics[width=0.4 \linewidth]{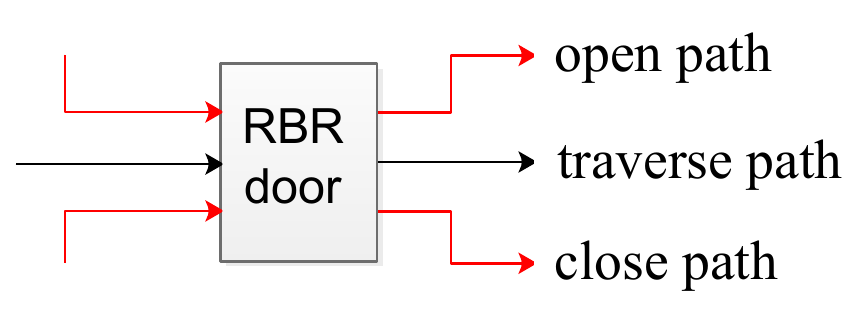}  
	\caption{RBR door gadget.}  
	\label{RBRdoor}   
\end{figure}

\begin{figure}[htbp]
	\centering  
	\includegraphics[width=0.4 \linewidth]{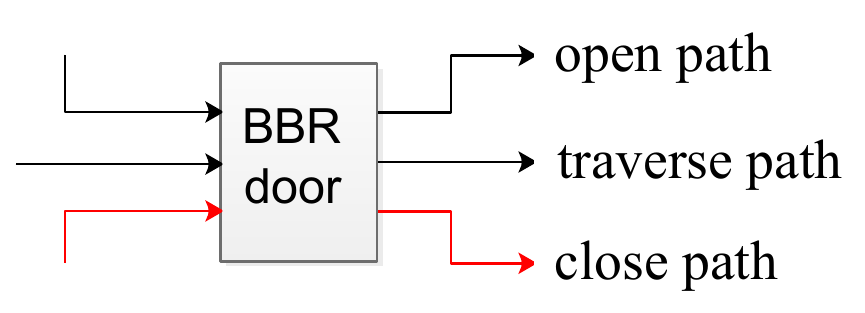}  
	\caption{BBR door gadget.}  
	\label{BBRdoor}   
\end{figure}

\textbf{BBR door gadget.} 
A BBR door gadget is an object with two states, open and closed.
The BBR door gadget contains three paths, open path, traverse path and close path.
The black avatar can traverse the traverse path if and only if the BBR gadget is in the open state.
When the black avatar traverses the open path, it is allowed to open the gadget.
When the red avatar traverses the close path, it is allowed to close the gadget.
Moreover, there is no leakage between three paths.

\begin{figure}[htbp]
	\centering  
	\includegraphics[width=1.0 \linewidth]{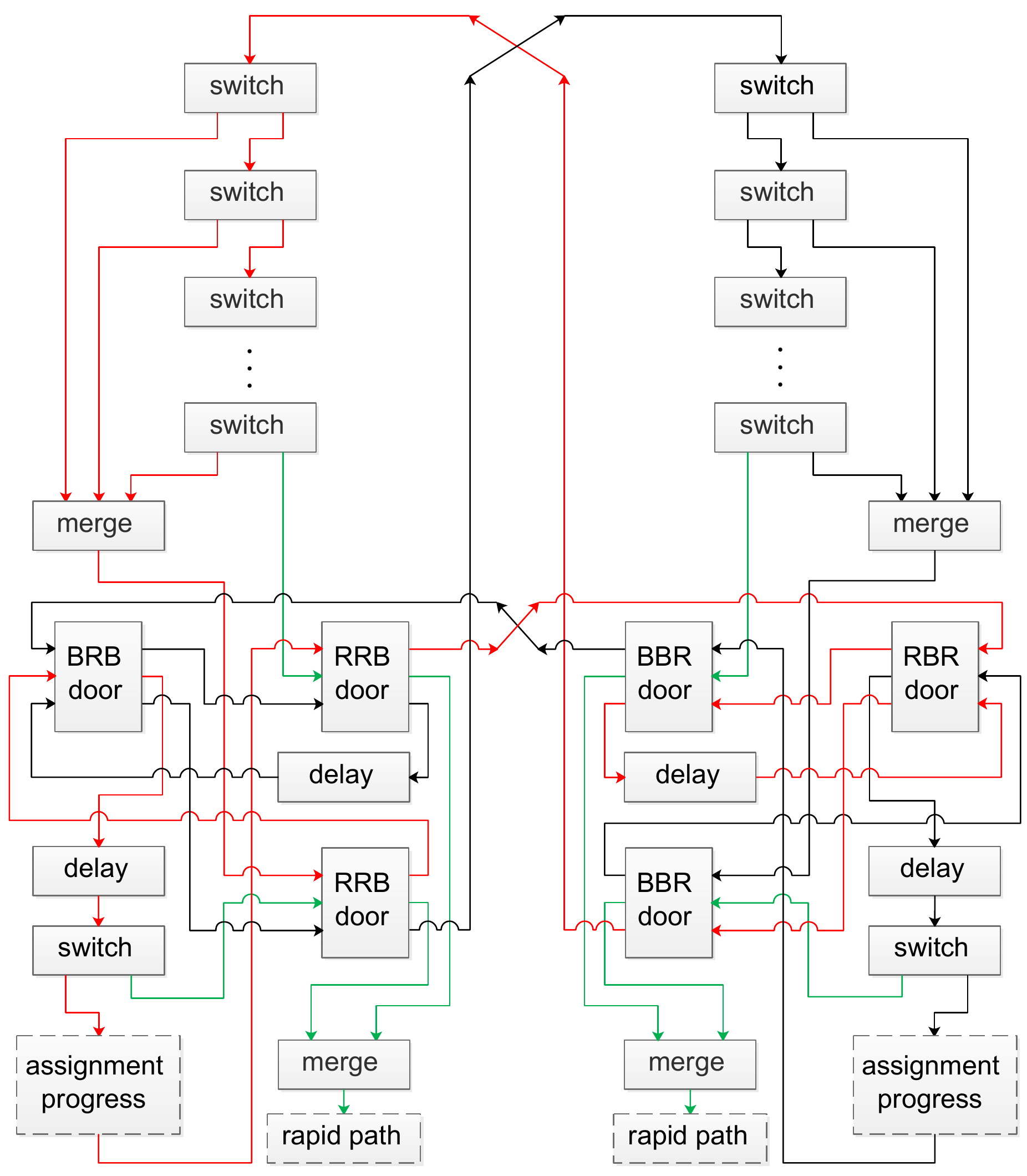}  
	\caption{Synchronization gadget.}  
	\label{synchronization}   
\end{figure}

\begin{figure}[htbp]
	\centering  
	\includegraphics[width=1.0 \linewidth]{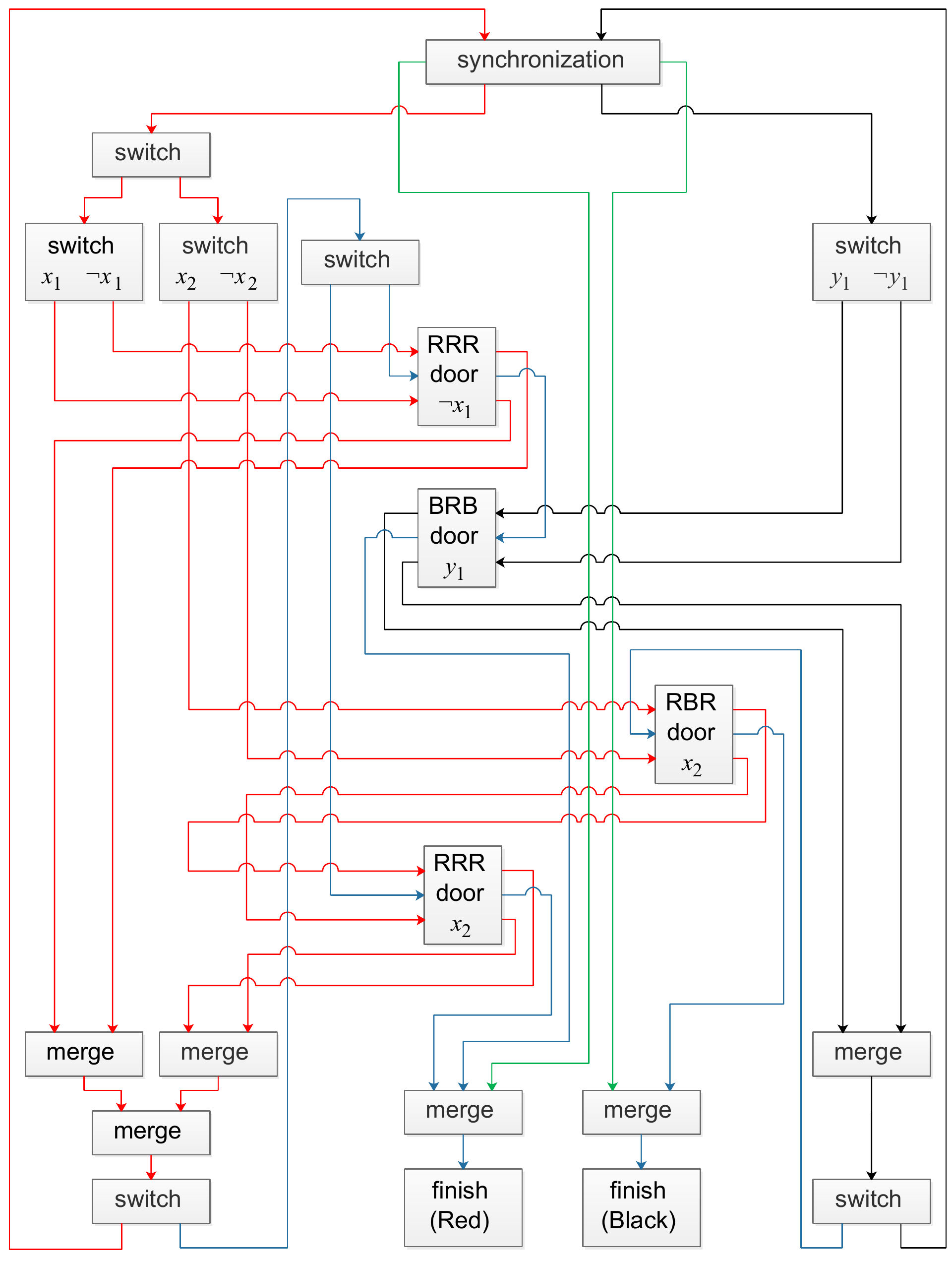}  
	\caption{EXPTIME-hardness framework.}  
	\label{EXPTIMEHframework}   
\end{figure}

\textbf{Synchronization gadget.} 
A synchronization gadget contains two paths for Red and Black respectively.
If the red avatar enters the synchronization gadget, the black avatar has to enter the gadget in given steps, otherwise, the red avatar can enter rapid path so that it can arrive the finish gadget swiftly; after the black enters the synchronization gadget, the red avatar has to leave the gadget, and it must enter the gadget again in given steps, and vice versa.
This progress simulate alternate move of two players in $G_2$.

The synchronization gadget could be constructed by other gadgets, and it is illustrated in Figure \ref{synchronization}.
At the beginning, the up RRB door gadget is in open state, while the BRB door, the down RRB door, the up BBR door and, the down BBR door and the RBR door gadgets are in closed state.
We assume that the black avatar enters assignment progress while the red avatar enters the top switch gadget.
There are so many switch gadgets for the red avatar to traverse that the black avatar has sufficient time to finish the assignment progress.
The black avatar should close the up RRB door gadget in certain steps, otherwise, the red avatar can enter rapid path (coloured in green).
Thus the black avatar should finish the assignment progress as soon as possible, while it has to open the up BBR door and the BRB door gadget before it closes the up RRB door gadget.
On the other hand, the red avatar should enter the merge gadget when the up RRB door gadget is closed by its opponent, otherwise, the red avatar will be stuck in the synchronization gadget, if the BRB door gadget is closed by the black avatar in certain steps.
Thus the red avatar has to open the down RRB door gadget, and it should traverse the BRB door gadget as soon as possible.
Then the red avatar can enter a switch gadget connecting to the down RRB gadget in certain steps.
This forces the black avatar to close the down RRB gadget, since the red avatar can enter rapid path through the door immediately.
Then the red avatar enters the assignment progress, and the black avatar enters the top switch gadget.
Next, the situation is symmetric.

Again, we use a example to describe the hardness framework. 
Figure \ref{EXPTIMEHframework} illustrates the framework corresponding to the formula game where
$X=\{x_1, x_2 \}$, $Y=\{y_1\}$, $\text{I-WIN} = (\neg x_1 \wedge y_1) \vee x_2$ and $\text{II-WIN} = x_2$.
Notice, I-WIN and II-WIN in the example are not formulas in 12DNF, however, it is easy to construct framework corresponding to the instance of $G_2$ where I-WIN and II-WIN are 12DNF via this example.
At the beginning, the red avatar leaves the synchronization gadget, and the black avatar enters the synchronization gadget.
Then the red avatar has to finish the assignment progress, so it has to choose one of two switch gadgets corresponding to variables $x_1$ and $x_2$.
Suppose that the red avatar choose the switch gadget corresponding to $x_2$, and it leaves the switch gadget through the right exit.
This indicates that player I assigns variable $x_2$ false in the formula game.
Then, the red avatar has to close the RBR door and the RRR door gadgets corresponding to $x_2$.
After operating door gadgets, the red avatar enters merge gadgets and a switch gadget.
Now, the red avatar can choose between entering the synchronization gadget or entering check progress (coloured in blue).
The red avatar can traverse the check path if and only if the formula $\text{I-WIN} = (\neg x_1 \wedge y_1) \vee x_2$ is true under current assignment.
If the red avatar can not enter check progress, it has to move back to the synchronization gadget, then the black avatar leaves the synchronization gadget, which indicates that player II begins to assign variables in $Y$.

It is obvious that the player I has a forced win in the formula game if and only if the red avatar arrives the finish gadget before its opponent does.
Moreover, number of gadgets in the framework is polynomial, so that the reduction could be established in polynomial time.

\textbf{Remarks.}
Similar to the open-close (RRR) door gadget, other door gadgets may contain additional paths connecting to finish gadgets, and we will see these types of gadgets in Section 4 and 5.
Formula games $G_1$ through $G_6$ described in \cite{difficultgames} are all EXPTIME-complete, and we can establish frameworks corresponding to these games by similar method.

\newpage

\section{Complexity of Janggi}

Before we discuss complexity of Xiangqi and Janggi, we introduce the notations used in this note.
Pieces table of Xiangqi and Janggi are illustrated in Figure \ref{piecestable}.
We use two capital letters to represent one piece in this note, for example, ``RG'' means ``red general''.
For convenience, we use pieces of same size, while pieces of Janggi differ in size according to their rank usually.
Pieces of Janggi are normally coloured in red and green, but we use red and black pieces here.
Especially, to avoid dazzle, we use pieces without Chinese characters to represent unmoveable rooks in figures.

\begin{figure}[htbp]
	\centering  
	\includegraphics[width=0.7 \linewidth]{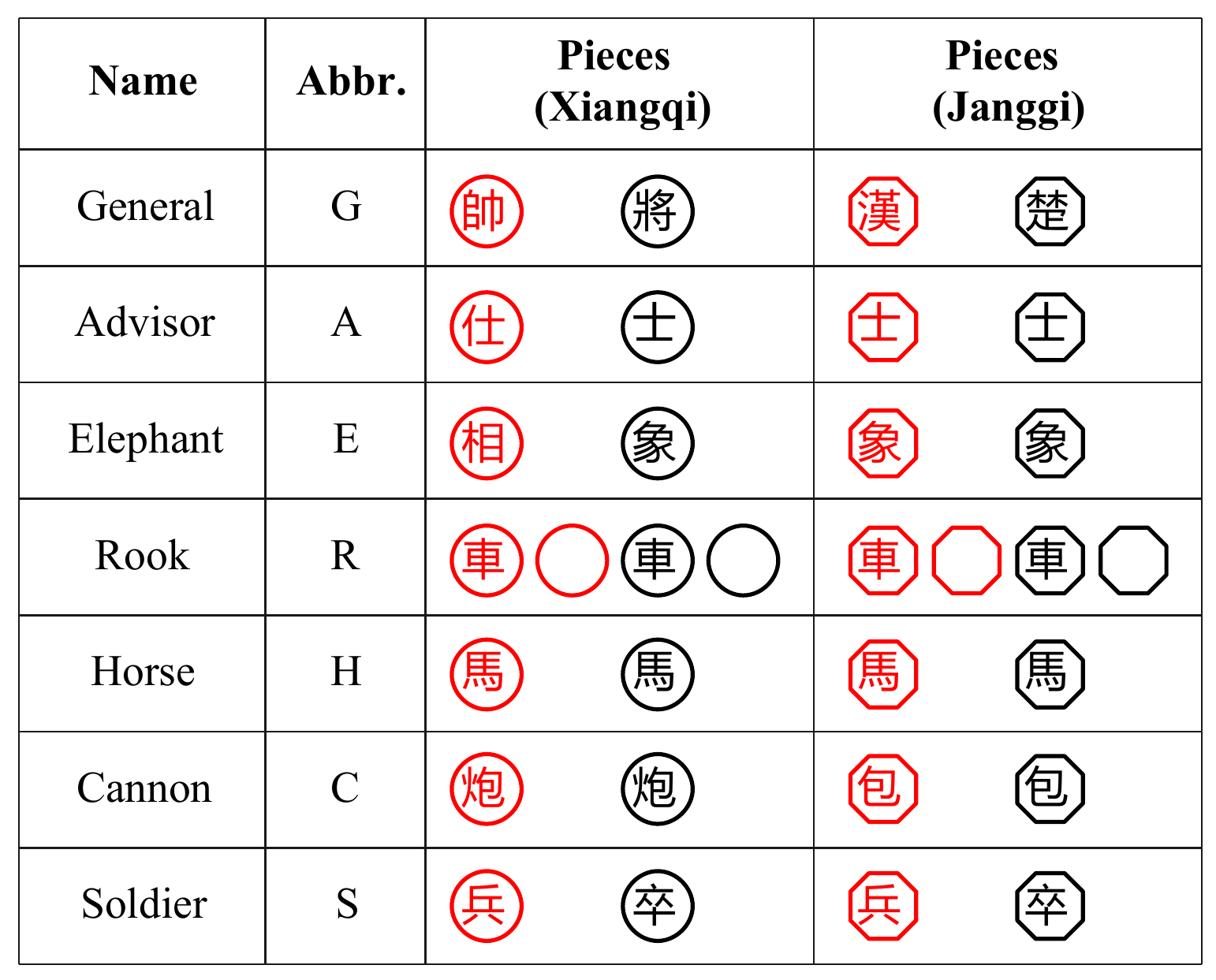}  
	\caption{Pieces table of Xiangqi and Janggi.}  
	\label{piecestable}   
\end{figure}

It is easier to construct gadgets in Janggi than in Xiangqi, so we discuss complexity of Janggi first.
Decision problem of Janggi is to decide whether Red has a forced win in a given position.
In our reduction, Red and Black control cannons which indicates the avatars in hardness frameworks, and we call these two cannons main cannons.
In all gadgets of Janggi, main cannons traverse the paths which are composed of elephants, and the elephants are protected by a large number of unmoveable rooks so that main cannon can not capture these pieces.

\subsection{NP-hardness of Janggi}

To prove NP-hardness of Janggi, we need to construct all gadgets of NP-hardness framework in Janggi.

\textbf{Start gadget.} 
Figure \ref{Jstart} illustrates a start gadget of Janggi. 
In this gadget, only main RC could move, and it can move east to leave the start gadget. 
If main RC captures any black pieces, it will be captured by one of the BR's immediately. 

\begin{figure}
	\begin{minipage}[t]{0.5\linewidth}
		\centering
		\includegraphics[width=0.7 \linewidth]{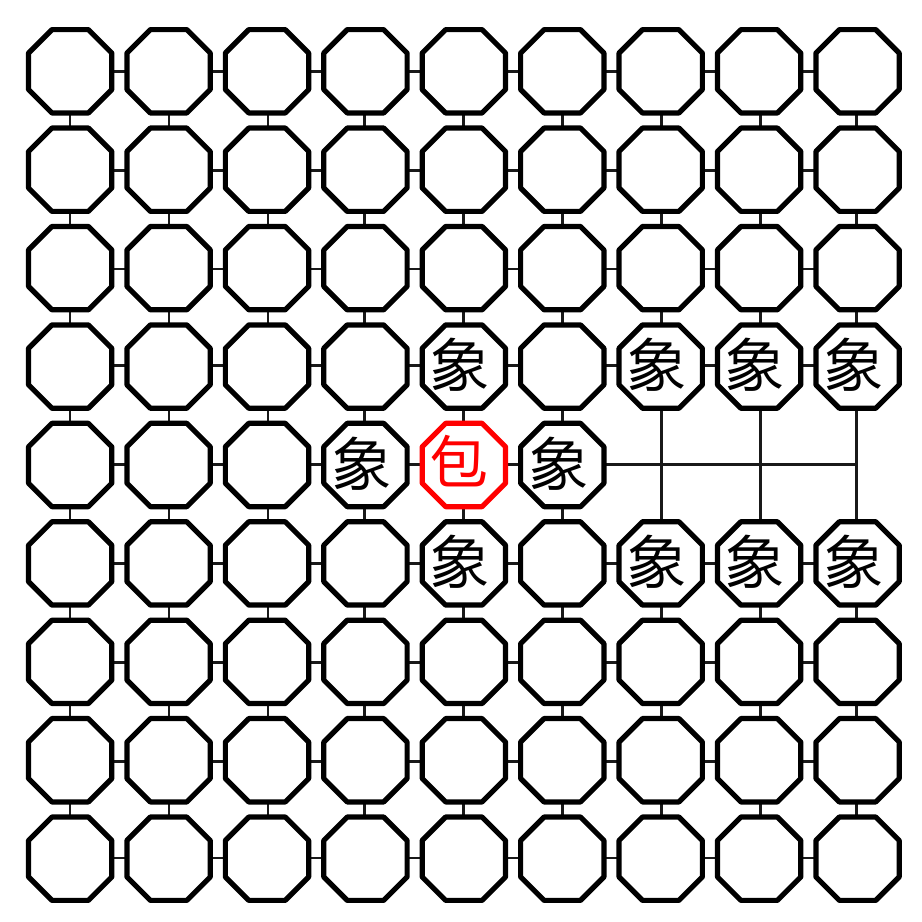}
		\caption{Start gadget of Janggi.}
		\label{Jstart}
	\end{minipage}%
	\begin{minipage}[t]{0.5\linewidth}
		\centering
		\includegraphics[width=0.7 \linewidth]{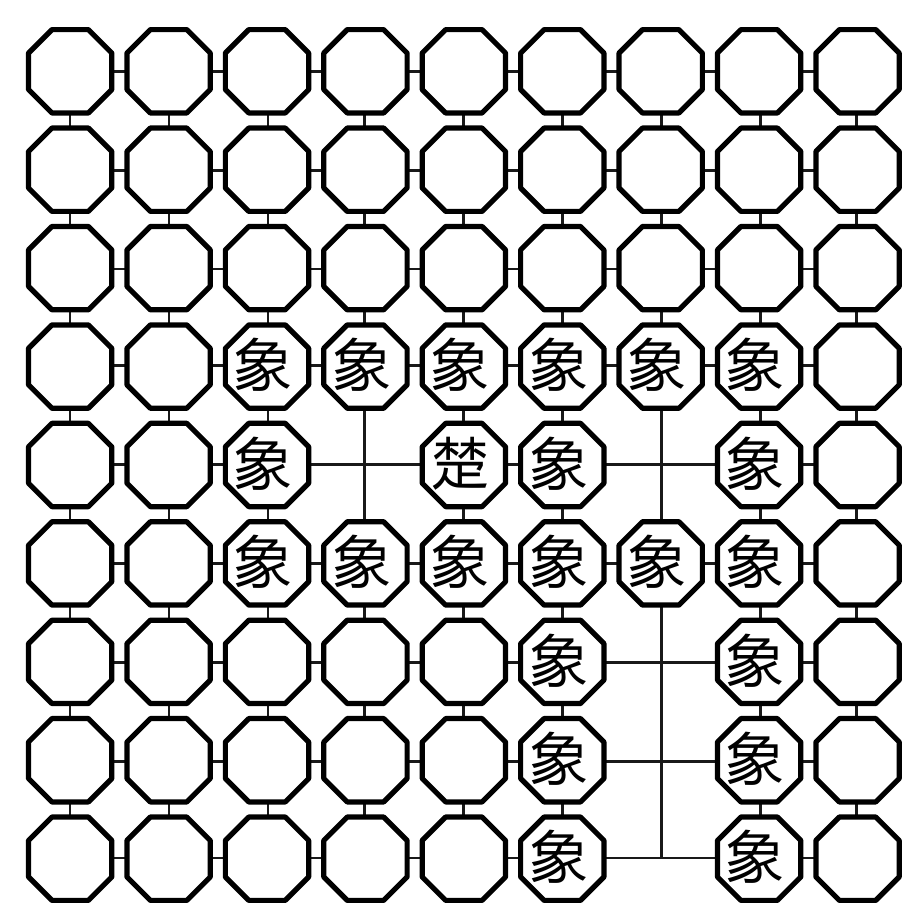}
		\caption{Finish gadget of Janggi.}
		\label{Jfinish}
	\end{minipage}
\end{figure}

\textbf{Finish gadget.} 
Figure \ref{Jfinish} illustrates a finish gadget of Janggi. 
Once main RC enters a finish gadget from south, the BG will be checkmated immediately. 
We assume that, in a finish gadget, the BG is at the corner of ``palace'', and other gadgets are all away from palace. 
Thus, moves along the diagonal lines in palace could be ignored.

\textbf{Turn gadget.} 
Figure \ref{Jturn} illustrates a turn gadget of Janggi. 
Main RC can enter a turn gadget from north, and it can move east to leave. 

\begin{figure}
	\begin{minipage}[t]{0.5\linewidth}
		\centering
		\includegraphics[width=0.7 \linewidth]{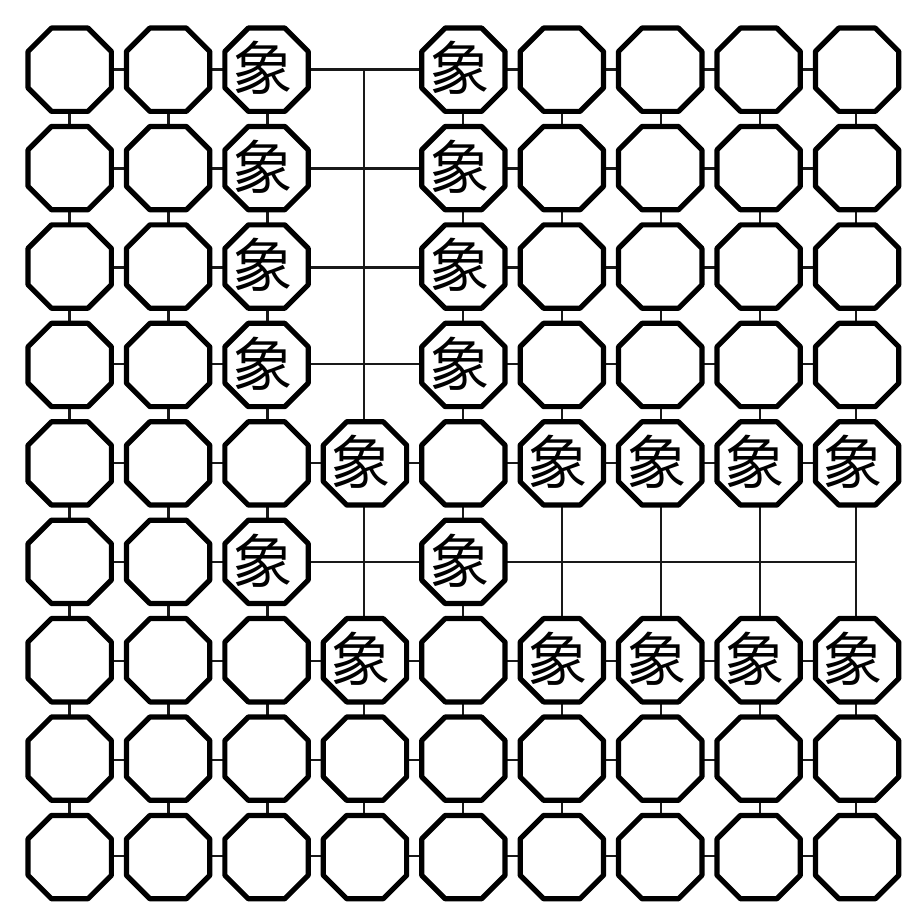}
		\caption{Turn gadget of Janggi.}
		\label{Jturn}
	\end{minipage}%
	\begin{minipage}[t]{0.5\linewidth}
		\centering
		\includegraphics[width=0.7 \linewidth]{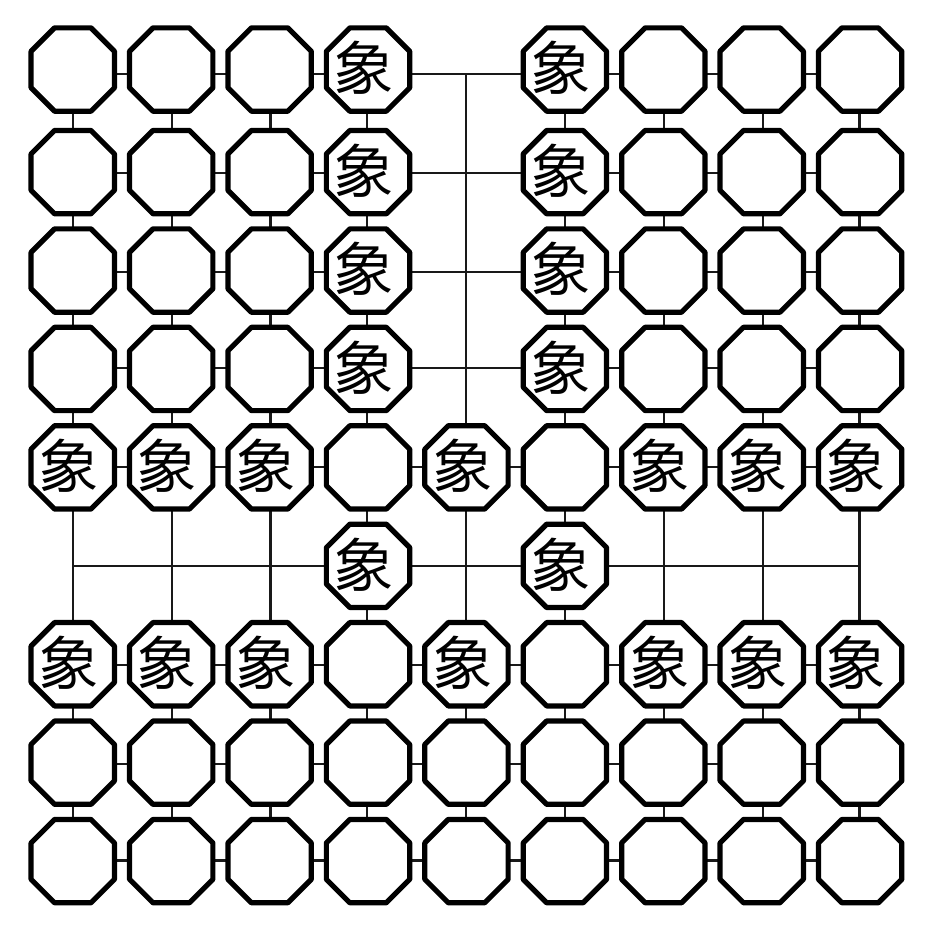}
		\caption{Switch and merge gadgets of Janggi.}
		\label{Jswitch}
	\end{minipage}
\end{figure}

\textbf{Switch and merge gadgets.} 
Figure \ref{Jswitch} illustrates a switch gadget of Janggi. 
When main RC enters a switch gadget from north, it can move east or west to leave. 
Moreover, a merge gadget is identical to a switch gadget in Janggi.

\textbf{One-way gadget.} 
Figure \ref{Joneway} illustrates an one-way gadget of Janggi. 
In this gadget, only one BR can move. 
Main RC can only traverse the gadget from north to south. 
When main RC enters an one-way gadget from north, it forces the BR to stay where it is, otherwise, the BR will be captured by main RC.
Then, main RC can move east and south to leave the gadget.
If main RC enters an one-way gadget from south, the BR can move south to threaten main RC, so that main RC can not traverse the gadget reversely. 

\begin{figure}[htbp]
	\centering  
	\includegraphics[width=0.35 \linewidth]{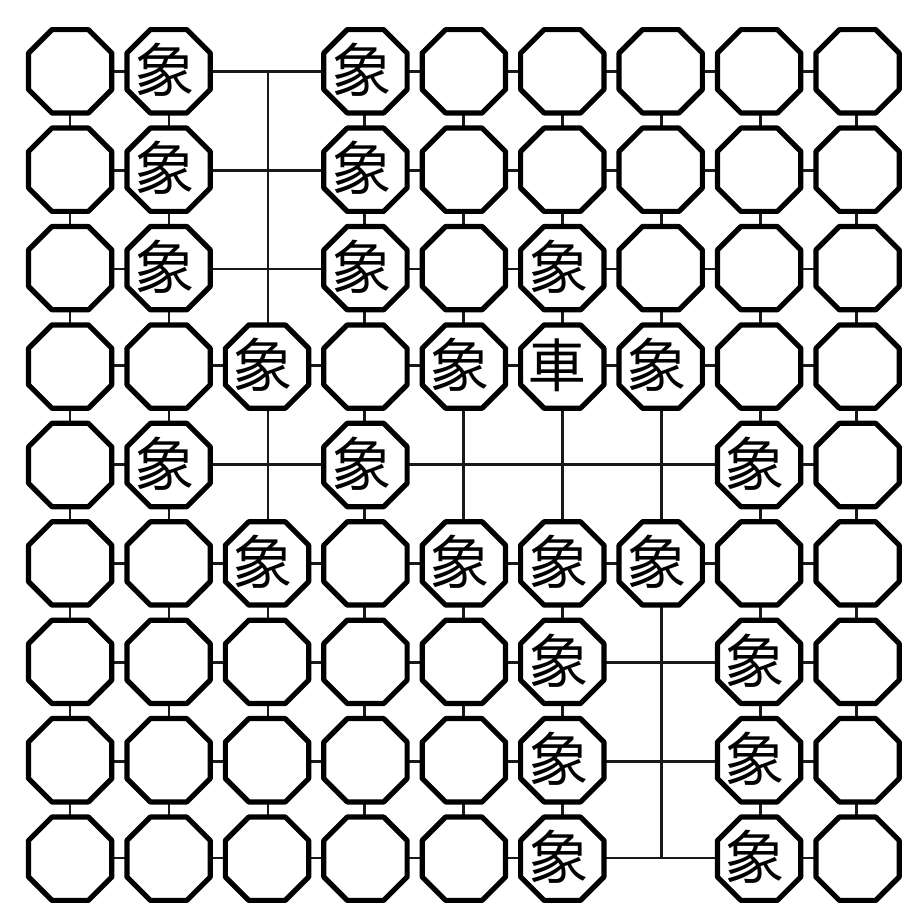}  
	\caption{One-way gadget of Janggi.}  
	\label{Joneway}   
\end{figure}

\begin{figure}[htbp]
	\centering  
	\includegraphics[width=0.47 \linewidth]{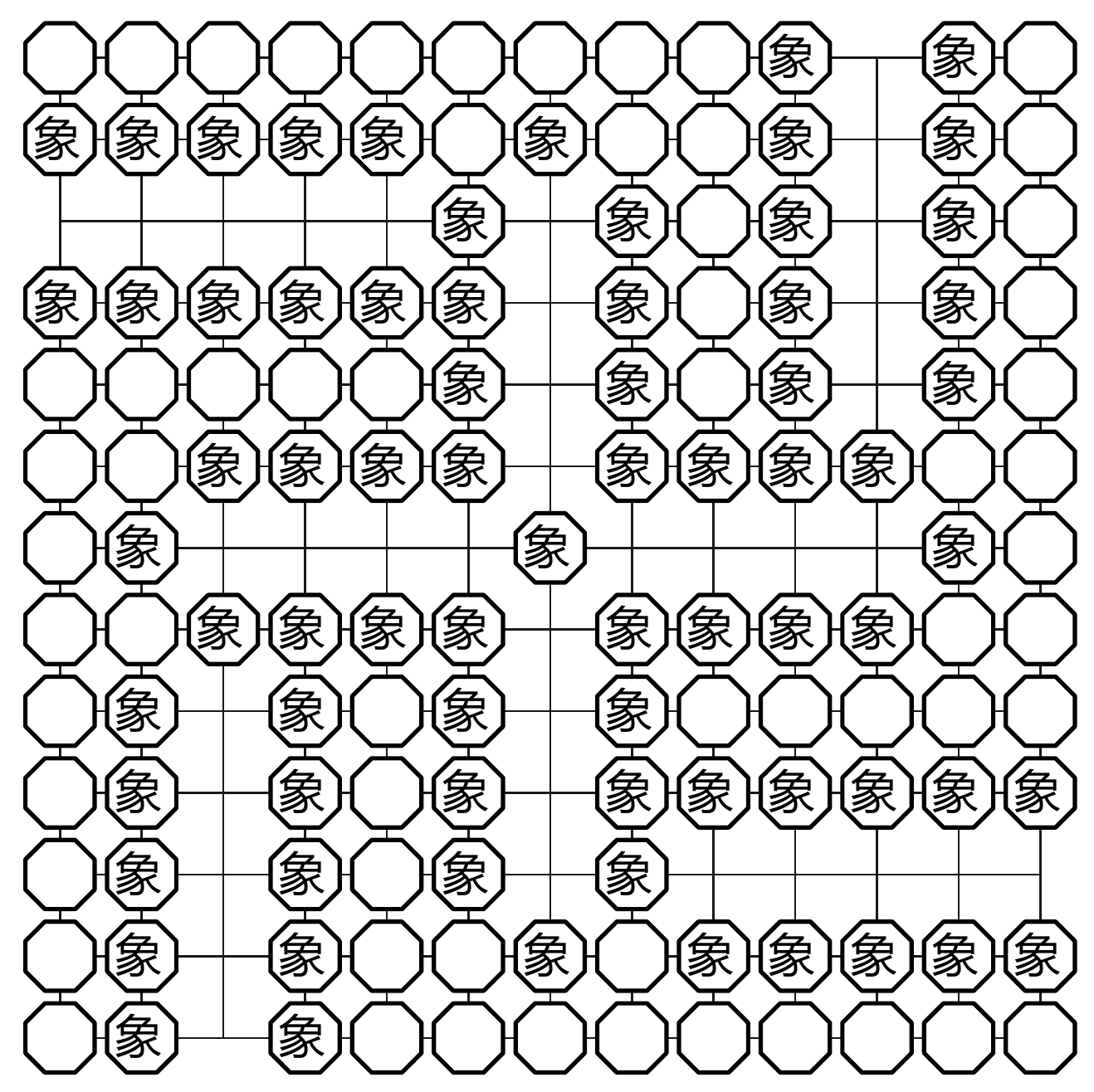}  
	\caption{Crossover gadget of Janggi.}  
	\label{Jcrossover}   
\end{figure}

\textbf{Crossover gadget.} 
Figure \ref{Jcrossover} illustrates a crossover gadget of Janggi. 
Main RC can traverse a crossover gadget through two paths, and there is no leakage between two paths of the gadget.
For one path, main RC enters the gadget from west, and then it moves south and east to leave. 
Situation of another path is similar. 
Notice, the BE at centre of the gadget will never be captured by main RC. 

\textbf{Door gadget.} 
Figure \ref{Jdoor} illustrates a door gadget of Janggi. 
In this gadget, only the BE at f6 can move.
Points d9 and a6 are the entrance and the exit of the open path respectively.
Path of point sequence (d9, d6, f6, d6, a6) is an open path.
Points i1 and n4 are the entrance and the exit of the traverse path respectively. 
Path of point sequence (i1, i4, n4) is a traverse path. 
In Figure \ref{Jdoor}, the gadget is in the closed state. 
Since the BE at f6 protects point i4, main RC can not traverse the traverse path. 
When main RC enters a door gadget from north, it can capture the BE at f6, which makes the door gadget in the open state.
If the BE at f6 moves to i4 in order to avoid being captured, main RC can directly move west to leave the gadget.
Since the BE can not move any more once it moves to i4.
When main RC traverses the traverse path, it still can capture the BE at i4.

\begin{figure}[htbp]
	\centering  
	\includegraphics[width=0.55 \linewidth]{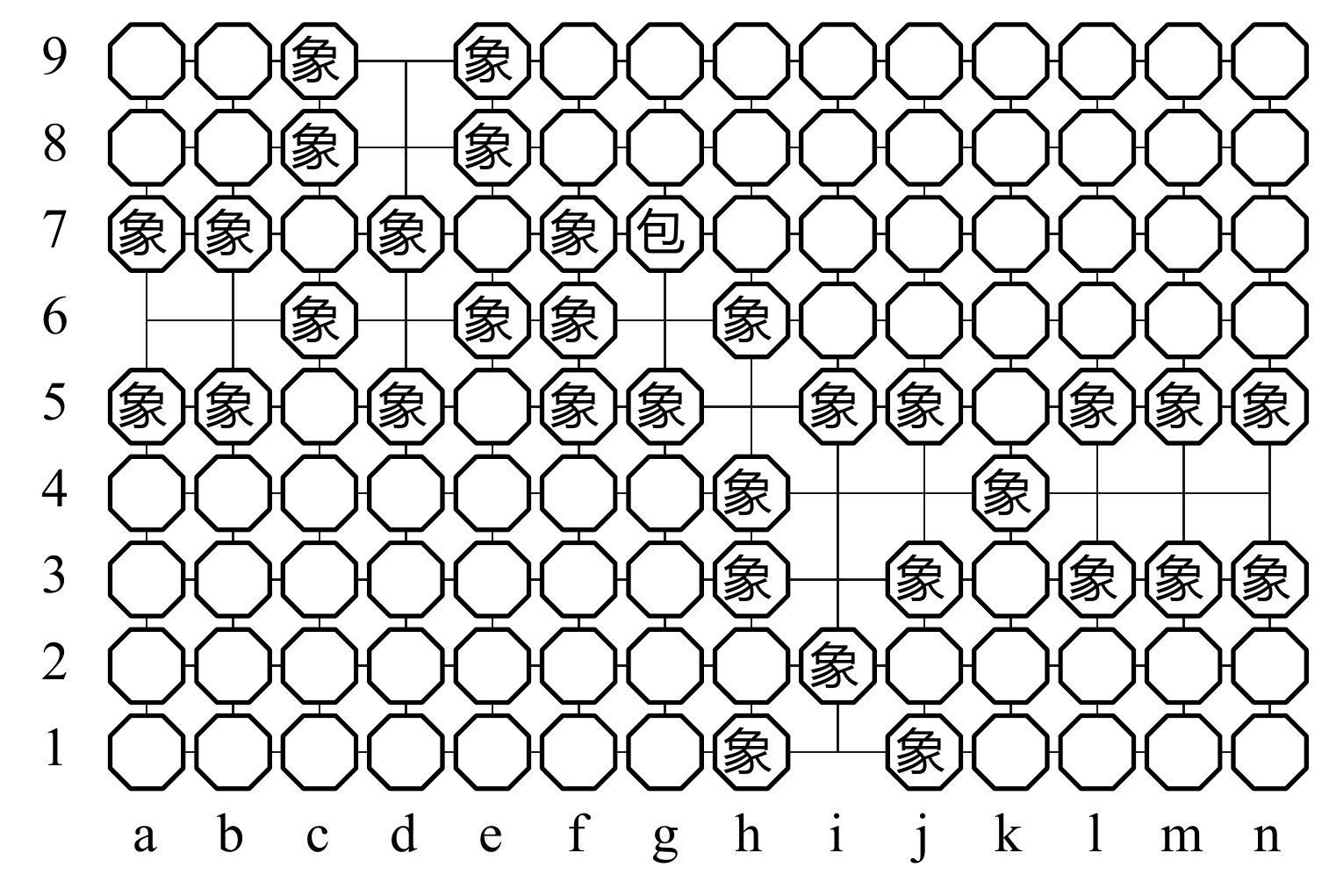}  
	\caption{Door gadget of Janggi.}  
	\label{Jdoor}   
\end{figure}

As all gadgets of NP-hardness framework have been constructed in Janggi, we obtain the following result.

\begin{proposition}
	Janggi is NP-hard.
\end{proposition}

\subsection{PSPACE-hardness of Janggi}

To prove PSPACE-hardness of Janggi, we need to construct all gadgets of PSPACE-hardness framework in Janggi.
Fortunately, the gadgets constructed in last section could be reused here.
So we just need to construct the open-close door gadget.

\textbf{Open-close (RRR) door gadget.} 
An open-close door gadget of Janggi is more complicated than gadgets above, and it is illustrated in Figure \ref{JRRRdoor}. 
In this gadget, only one RC (at q20), two RR's (at i29 and q7), one BC (at m20), one BR (at m21) and two BE's (at i19 and q21) can move. 
We call the moveable RC and BC control cannons.
When control RC stops at q20, the gadget is considered in the closed state; when control RC stops at i20, the gadget is considered in the open state.

\begin{figure}[p]
	\centering  
	\includegraphics[width=1.0 \linewidth]{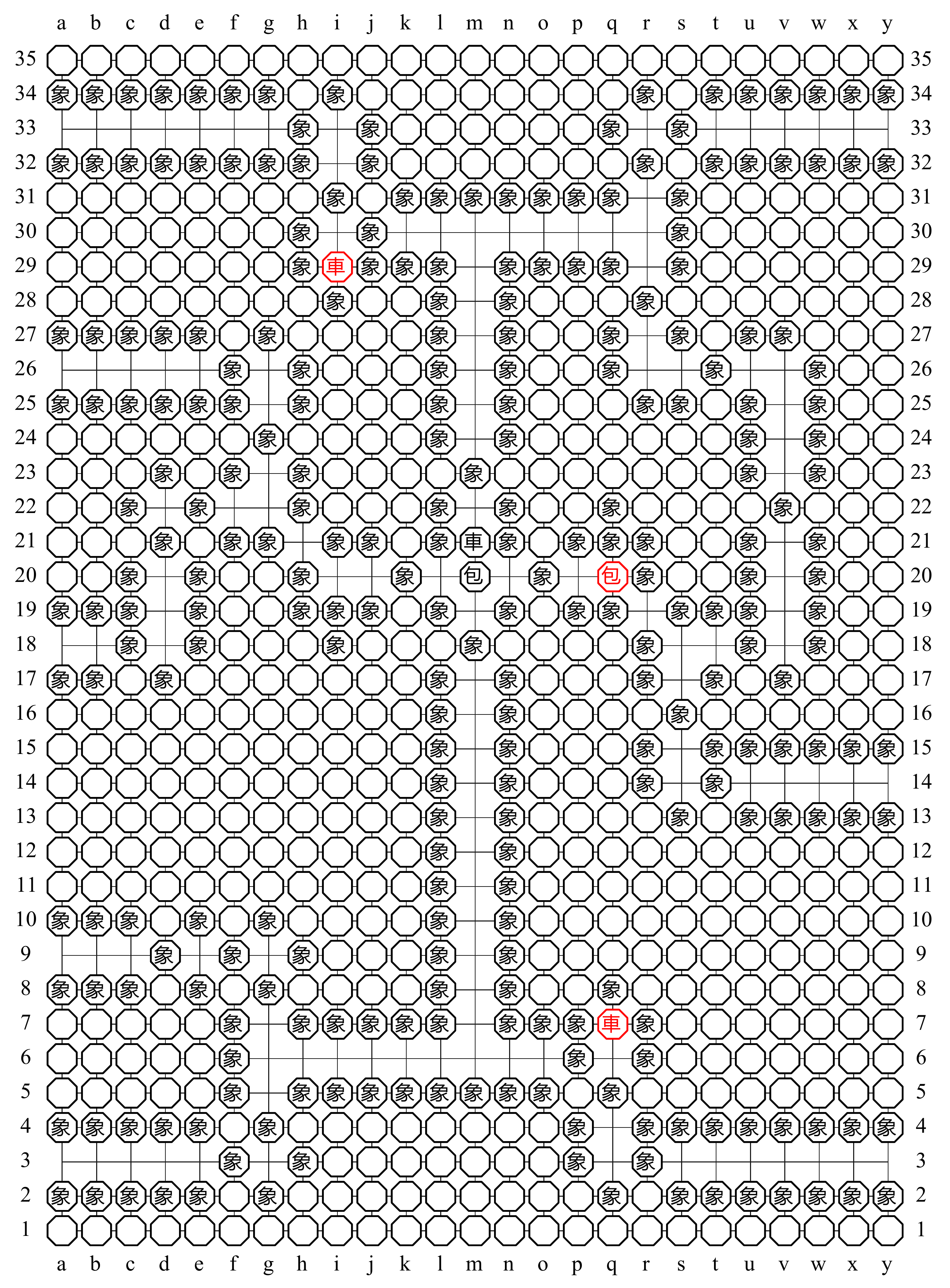}  
	\caption{Open-close (RRR) door gadget of Janggi.}  
	\label{JRRRdoor}   
\end{figure}

Points a3 and a9 are the entrance and the exit of the open path respectively. 
Path of point sequence (a3, g3, g6, g9, e9, a9) is an open path. 
Points a18 and a26 are the entrance and the exit of the traverse path respectively. 
Path of point sequence (a18, d18, d22, g22, g26, a26) is a traverse path. 
Points y33 and y14 are the entrance and the exit of the close path respectively. 
Path of point sequence (y33, r33, r30, r26, v26, v18, s18, s14, y14) is a close path.
Points a33 and y3 are entrances of rapid checkmate paths. 

Main RC can open the gadget by traversing the open path. 
When main RC moves to g6, control BC has to move to m6 to obstruct main RC. 
Otherwise, main RC can enter the rapid checkmate path. 
After control BC leaves m20, control RC can move to l20 and i20 to obstruct the BE at i19, which makes the gadget in the open state.
Then, when main RC moves to e9, control BC should move back to m20 to restrict movements of control RC.

Main RC can traverse the traverse path if and only if the gadget is in the open state.
If control RC stops at i20, main RC can traverse the path safely. 
If the gadget is in the closed state, main RC will be captured by the BE at i19 when it move to g22.

When main RC traverses the close path, it has to close the gadget.
To traverse the path, main RC must pass point s18, however, the point is protected by the BE at q21. 
Thus, after main RC moves to r30, and it forces control BC to move to m30, control RC must move to q20 to obstruct the BE, which makes the gadget in the closed state.
Then, when main RC moves to v18, control BC should move back to m20 to restrict movements of control RC.

Here, we consider some nontrivial improper moves in an open-close gadget:

(1) If control RC attempts to leave the gadget, it must pass point m20. 
However, the BR at m21 protects point m20. 

(2) If control BC attempts to leave the gadget, it must pass points i30 or q6. 
However, these two points are protected by two RR's (at i29 and q7).

(3) If the RR's (at i29 and q7) move to i30 or q6 without capturing black pieces.
These moves do not benefit Red, and Black can just ignore them.

(4) If the BR at m21 moves to m20 in order to obstruct control RC, control RC can stops at i20 or q20 to force the BR to leave m20.

(5) If the BE's (at i19 and q21) move to g22 or s18 without capturing main RC, the BE's will be captured when main RC traverses the paths.

(6) When main RC traverses the open path, it moves to e9, and control BC may move to m20. 
Then, if main RC moves back to g6, it forces control BC to move to m6 again. 
However, it is no benefit to Red, because Red requires two moves while Black just requires one move. 
The difference on number of moves also avoids repeating moves.

(7) When main RC traverses close path, if it moves back to r30 after arriving v18. 
The situation is similar to case (6).

(8) Main RC moves into m column. 
It implies that main RC has passed i30 or q6. 
However, now that main RC can arrive i30 or q6, it can enter rapid checkmate paths. 
Thus there is no need for main RC to move into m column.

(9) If control BC does not move back to m20 when main RC leaves the gadget.
Movements of control RC will not be restricted, then Red can open the gadget at any time.
Thus control BC should move back to m20 as far as possible.

As all gadgets of PSPACE-hardness framework have been constructed in Janggi, we obtain the following result.

\begin{proposition}
	Janggi is PSPACE-hard.
\end{proposition}

\subsection{EXPTIME-hardness of Janggi}

To prove EXPTIME-hardness of Janggi, we need to construct all gadgets of EXPTIME-hardness framework in Janggi.
The gadgets constructed in last two sections could be reused, and the start, finish, turn, switch, merge and one-way gadgets for Black could be constructed symmetrically.
We do not need to construct new crossover gadget, since three types of crossover gadgets are identical in Janggi.
Moreover, the RRR door gadget is identical to the open-close gadget in last section, and BBB door, BRB door and RRB door gadgets could be constructed symmetrically, so we just need to construct RBR door and BBR door gadgets.
Since paths are composed of red elephants and black elephants in our reduction, the border between unmoveable red rooks and black rooks must be taken into consideration, and it is illustrated in Figure \ref{Jborder}.

\begin{figure}[htbp]
	\centering  
	\includegraphics[width=0.35 \linewidth]{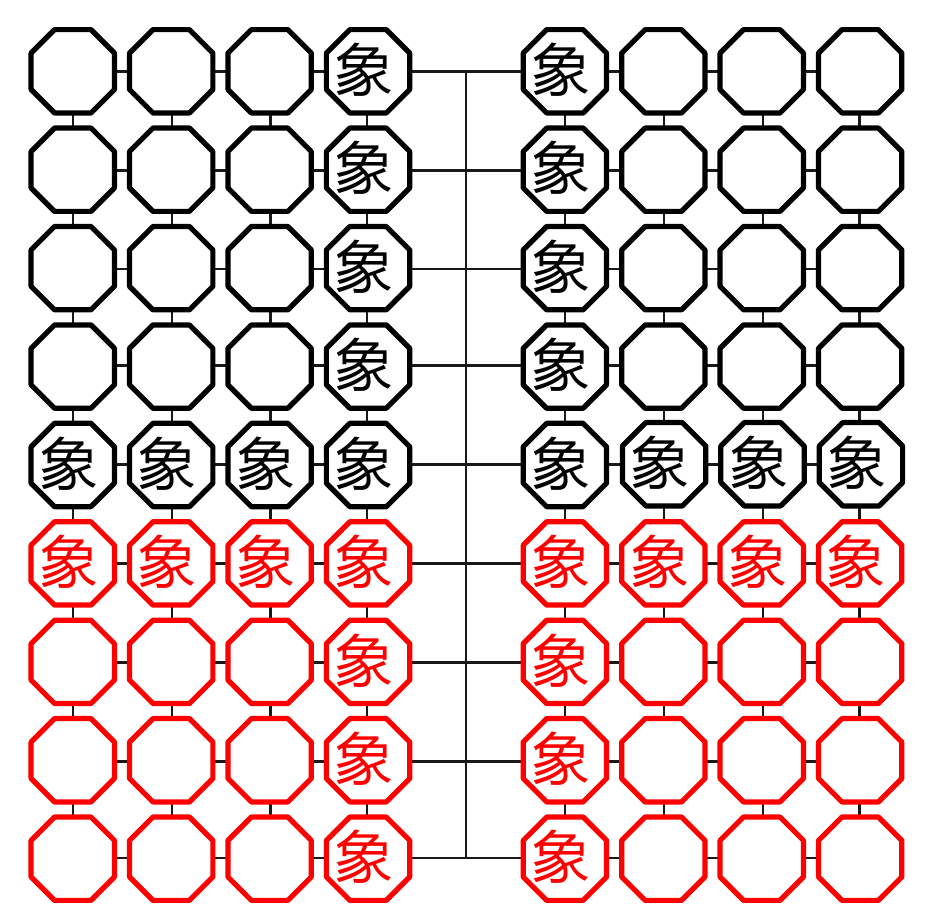}  
	\caption{Border between Red and Black in Janggi.}  
	\label{Jborder}   
\end{figure}

\textbf{RBR door gadget.} 
Figure \ref{JRBRdoor} illustrates a RBR door gadget of Janggi. 
In this gadget, only one RC (at f20), two RR's (at i29 and q7), one RE (at c18), one BC (at m20), three BR's (at g14, g23 and m21) and one BE (at q21) can move. 
We call the moveable RC and BC control cannons.
When control RC stops at f20, the gadget is considered in the closed state; when control RC stops at q20, the gadget is considered in the open state.

Points y33 and y14 are the entrance and the exit of the open path respectively. 
Path of point sequence (y33, r33, r30, r26, v26, v18, s18, s14, y14) is an open path for Red. 
Points a14 and a23 are the entrance and the exit of the traverse path respectively.
Path of point sequence (a14, f14, f21, f23, a23) is a traverse path for Black. 
Points a3 and a9 are the entrance and the exit of the close path respectively.
Path of point sequence (a3, g3, g6, g9, e9, a9) is a close path for Red.
Points a33 and y3 are entrances of rapid checkmate paths for Red.

When main RC traverses the open path, it has to open the gadget.
To traverse the path, main RC must pass point s18, however, the point is protected by the BE at q21. 
Thus, after main RC moves to r30, and it forces control BC to move to m30, control RC must move to q20 to obstruct the BE, which makes the gadget in the open state.
Then, when main RC moves to v18, control BC should move back to m20 to restrict movements of control RC.

Main BC can traverse the traverse path if and only if the gadget is in the open state.
If control RC stops at q20, main BC can traverse the path successfully. 
If the gadget is in the closed state, main BC can not traverse the path since a cannon can neither capture cannons nor ``jump'' over cannons in Janggi.

Main RC can close the gadget by traversing the close path. 
When main RC moves to g6, control BC has to move to m6 to obstruct main RC. 
Otherwise, main RC can enter the rapid checkmate path. 
After control BC leaves m20, control RC can move to f20 to block the traverse path, which makes the gadget in the closed state.
Then, when main RC moves to e9, control BC should move back to m20 to restrict movements of control RC.

\begin{figure}[htbp]
	\centering  
	\includegraphics[width=1.0 \linewidth]{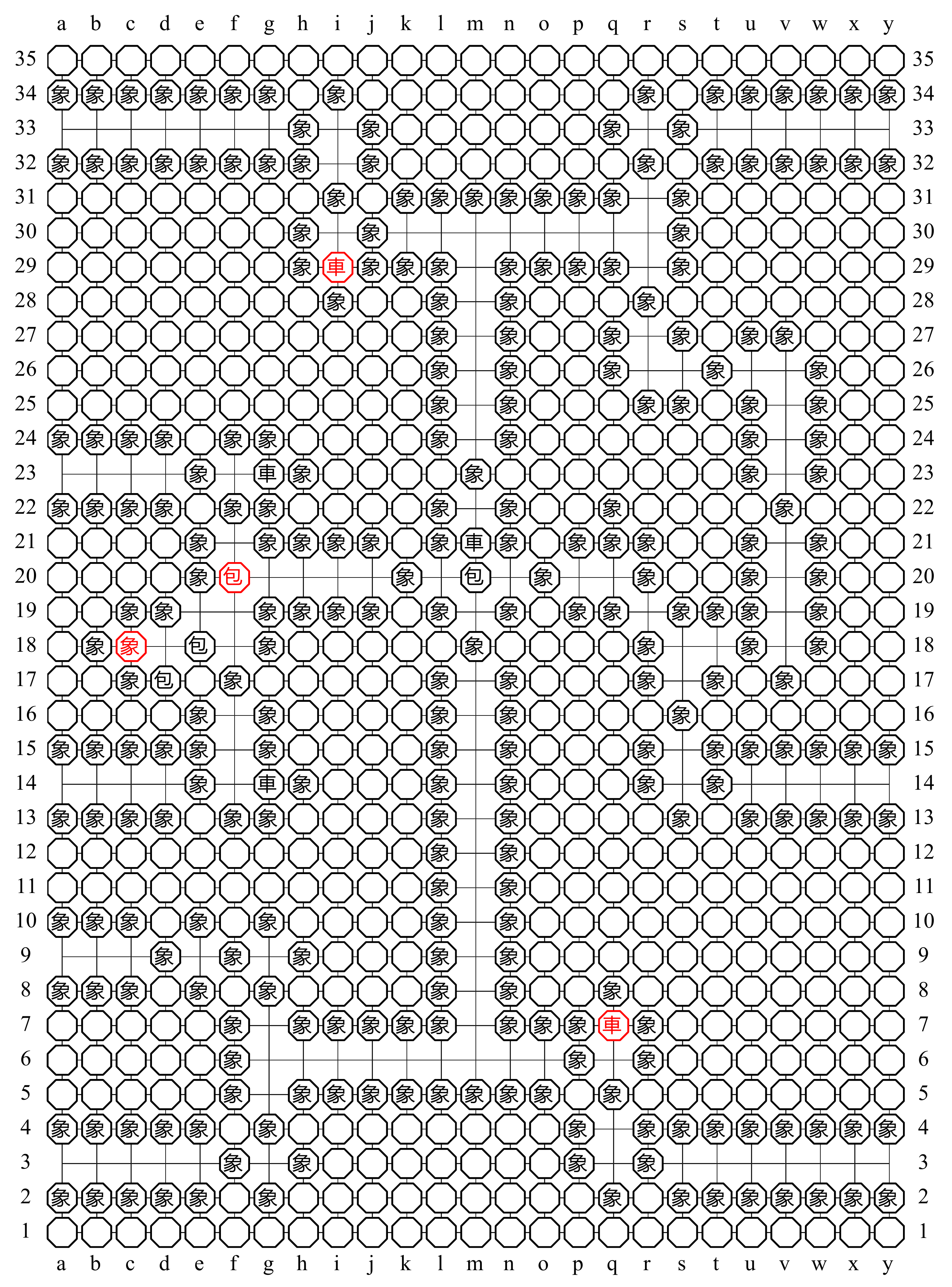}  
	\caption{RBR door gadget of Janggi.}  
	\label{JRBRdoor}   
\end{figure}

\begin{figure}[htbp]
	\centering  
	\includegraphics[width=1.0 \linewidth]{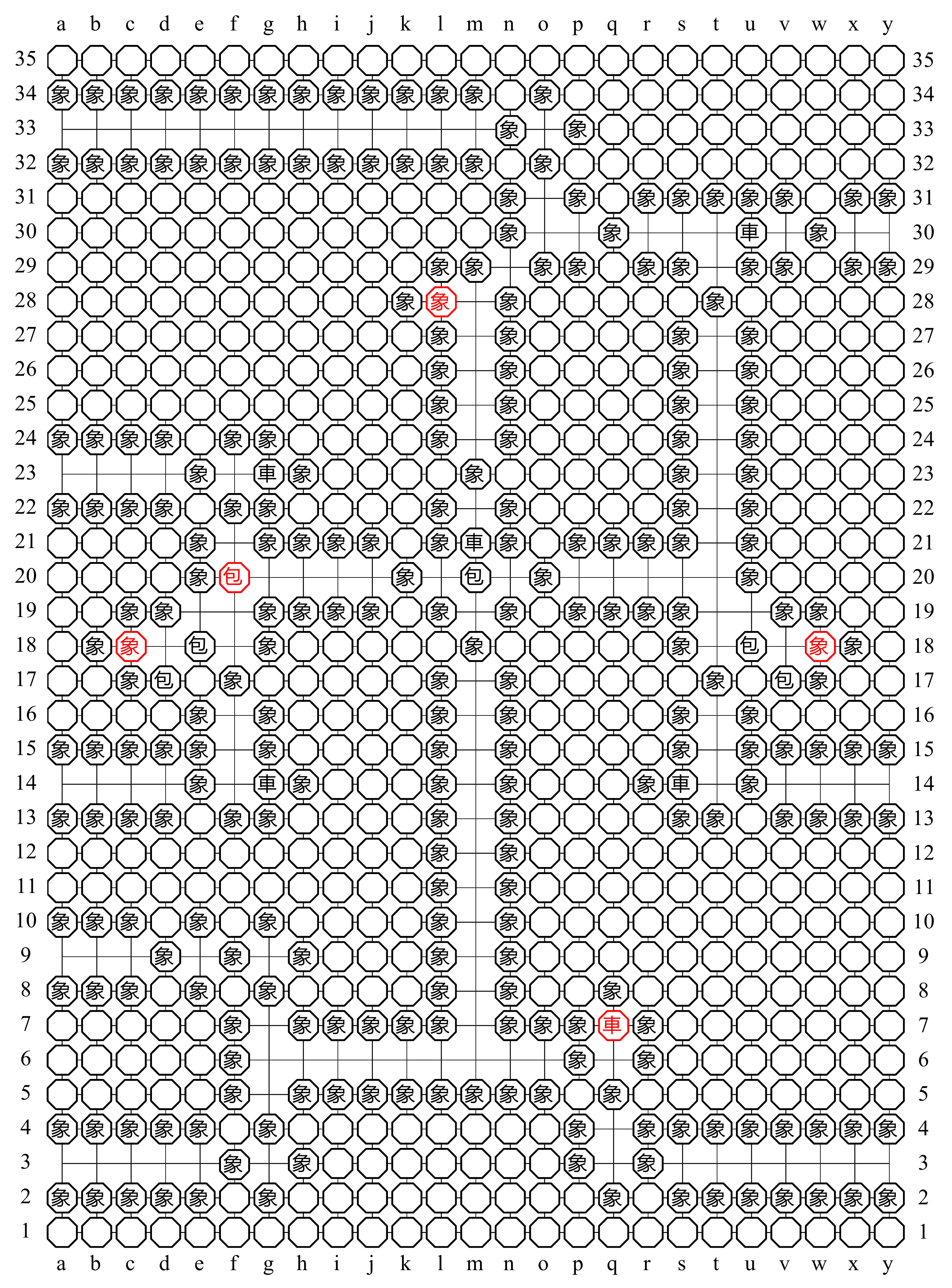}  
	\caption{BBR door gadget of Janggi.}  
	\label{JBBRdoor}   
\end{figure}

Here, we consider some different nontrivial improper moves in a RBR door gadget:

(1) Control RC can not leave the gadget since three BR's (at g14, g23 and m21) protect points f14, f23 and m20.

(2) If the RE at c18 moves to f20 without capturing one BC, main BC can capture the RE when it traverses the traverse path.

(3) Unless two BR's (at g14 and g23) can capture control RC, other moves of the BR's do not benefit Black.
Thus Red can just ignore moves of the BR's.

(4) When the gadget is in the open state, control BC moves to f20.
Then the RE at c18 can capture control BC, and control RC can move to j20 to prevent main BC from moving to m20.
The next time main RC traverse the open path or close path, it can enter rapid checkmate paths as no black piece can obstruct it.

(5) When the gadget is in the open state, main BC moves to f20.
The RE at c18 can capture main BC, then control BC has to capture the RE.
The situation is similar to case (4).

\textbf{BBR door gadget.} 
Figure \ref{JBBRdoor} illustrates a BBR door gadget of Janggi. 
In this gadget, only one RC (at f20), one RR (at q7), three RE's (at c18, l28 and w18), one BC (at m20) and five BR's (at g14, g23, m21, s14 and u30) could move. 
We call the moveable RC and BC control cannons.
When control RC stops at f20, the gadget is considered  in the closed state; when control RC stops at t20, the gadget is considered in the open state.

Points a33 and y30 are the entrance and the exit of the open path respectively. 
Path of point sequence (a33, o33, o30, t30, v30, y30) is an open path for Black. 
Points a14 and a23 are the entrance and the exit of the traverse path respectively. 
Path of point sequence (a14, f14, f21, f23, a23) is a traverse path for Black. 
Points a3 and a9 are the entrance and the exit of the close path respectively. 
Path of point sequence (a3, g3, g6, g9, e9, a9) is a close path for Red.
Point y3 is the entrance of rapid checkmate path for Red, and point y14 is the entrance of rapid checkmate path for Black.

Main BC can open the gadget by traversing the open path.
To traverse the path, main BC must pass point o30, however, the point is protected by the RE at l28. 
Thus control BC should move to m28 to obstruct the RE.
Once main BC enters the gadget, and control BC leaves m20, control RC has to move to n20, and it should stop at t20.
Otherwise, main BC can enter the rapid checkmate path.
Then, after main BC moves to t30, control BC should move back to m20 to restrict movements of control RC.

The traverse path in a BBR door gadget is similar to the path in a RBR door gadget.
Main BC can traverse the traverse path if and only if the gadget is in the open state, since a cannon can neither capture cannons nor ``jump'' over cannons in Janggi.

The close path in a BBR door gadget is also similar to the path in a RBR door gadget.
Main RC can close the gadget by traversing the close path.

Here, we consider some different nontrivial improper moves in a BBR door gadget:

(1) If control RC moves to t14 or t30 in order to leave the gadget, it will be captured by one of the BR's (at s14 and u30).

(2) If the RE at l28 moves to o30 without capturing main BC, main BC can capture the RE when it traverse the open path.

(3) If the RE at w18 moves to t20 without capturing one BC, Black can just ignore the move.

(4) When the gadget is in the closed state, if control BC moves to t20 in order to leave the gadget. 
The RE at w18 can capture control BC, and the next time main RC traverse the close path, it can enter rapid checkmate paths as no black piece can obstruct it.
Since control BC has been captured, main BC can neither traverse the open path nor enter rapid checkmate path.

(5) When the gadget is in the closed state, if main BC moves to t20.
It implies that main BC can enter the rapid checkmate path. 
Thus there is no need for main BC to move to t20.

(6) If main RC stops at g6 for a long time, control BC has to stop at m6.
So that main BC can not traverse the open path.
However, in the synchronization gadget of EXPTIME-hardness framework, the black avatar always traverse the open path of the BBR door gadgets before the red avatar traverse the close path.
Moreover, the red avatar can not stop at BBR door gadgets forever, since it has to prevent its opponent from entering rapid path in the framework.

(7) If main BC stops at t30 for a long time, the main RC can not close the BBR door gadget.
We discuss it in two cases:
(i) if main BC stop at the up BBR door gadget in the synchronization gadget, then the red avatar (main RC) can enter rapid path through the up RRB door gadget; (ii) if main BC stop at the down BBR door gadget, the black avatar (main BC) will be stuck in the synchronization gadget when the RBR door gadget is closed by the red avatar.

As all gadgets of EXPTIME-hardness framework have been constructed in Janggi, we obtain the following result.

\begin{theorem}
	Janggi is EXPTIME-complete.
\end{theorem}

\begin{proof}
	Firstly, Janggi could be solved by a brute search algorithm in exponential time, thus Janggi is in EXPTIME. 
	Secondly, we have constructed all gadgets of EXPTIME-hardness framework in Janggi, thus Janggi is EXPTIME-hard.
\end{proof}

\textbf{Remarks.}
We do not need advisors, horses and soldiers in our reduction.
If we assume that the whole instance of Janggi is surrounded by unmoveable cannons, we can construct reduction even if the game board is infinite.
We do not discuss the reachability of the instance of Janggi here.

\newpage

\section{Complexity of Xiangqi}

Similar to complexity analysis of Janggi, decision problem of Xiangqi is to decide whether Red has a forced win in a given position.
Red and Black control cannons which indicates the avatars in hardness frameworks, and we call these two cannons main cannons.
In all gadgets of Xiangqi, main cannons traverse the paths which are composed of elephants and horses, and the elephants and horses are protected by a large number of unmoveable rooks so that main cannon can not capture these pieces.

\subsection{NP-hardness of Xiangqi}

To prove NP-hardness of Xiangqi, we need to construct all gadgets of NP-hardness framework in Xiangqi.

\textbf{Start gadget.} 
Figure \ref{Xstart} illustrates a start gadget of Xiangqi.
In this gadget, only main RC is moveable, and it can move east to leave a start gadget.

\begin{figure}[htbp]
	\begin{minipage}[htbp]{0.5\linewidth}
		\centering
		\includegraphics[width=0.7 \linewidth]{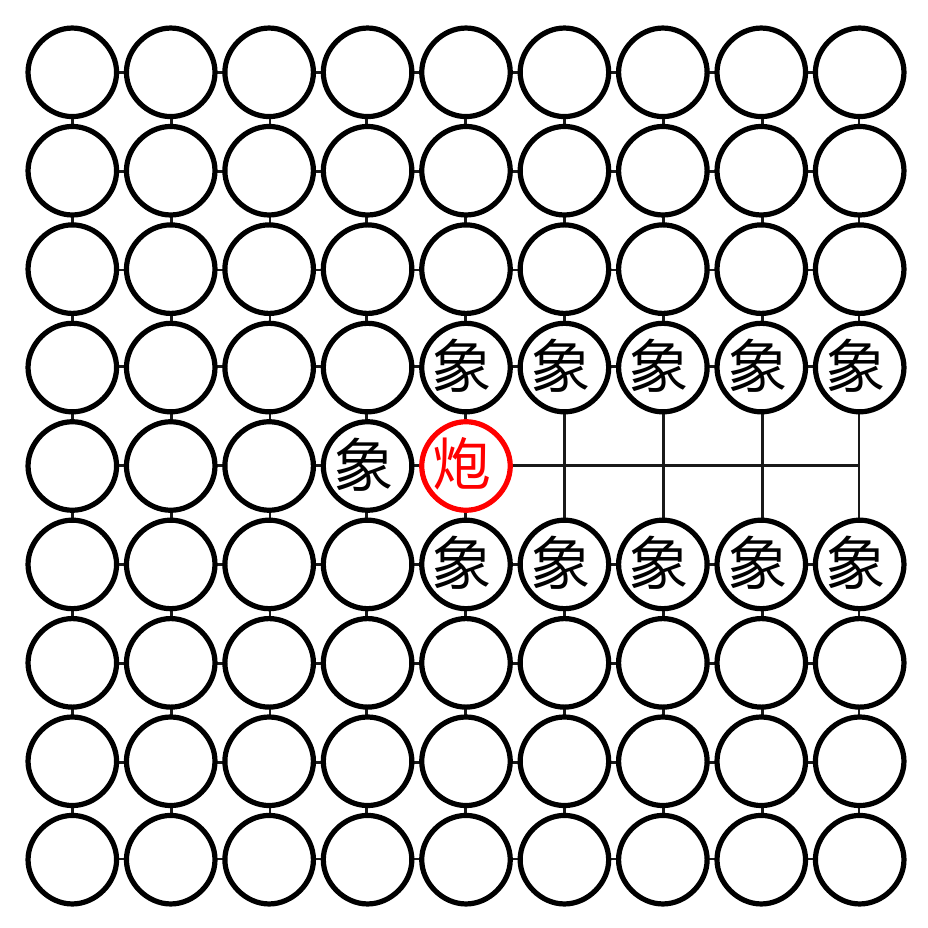}
		\caption{Start gadget of Xiangqi.}
		\label{Xstart}
	\end{minipage}%
	\begin{minipage}[htbp]{0.5\linewidth}
		\centering
		\includegraphics[width=0.7 \linewidth]{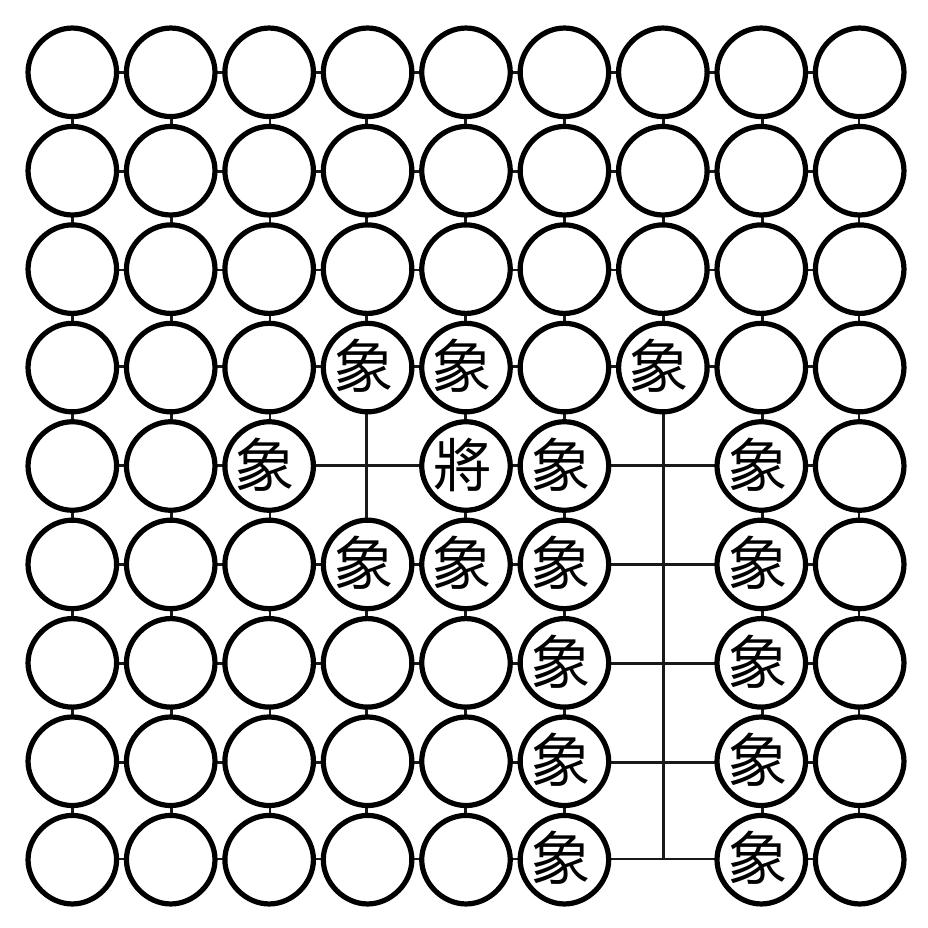}
		\caption{Finish gadget of Xiangqi.}
		\label{Xfinish}
	\end{minipage}
\end{figure}

\begin{figure}[htbp]
	\begin{minipage}[htbp]{0.5\linewidth}
		\centering
		\includegraphics[width=0.7 \linewidth]{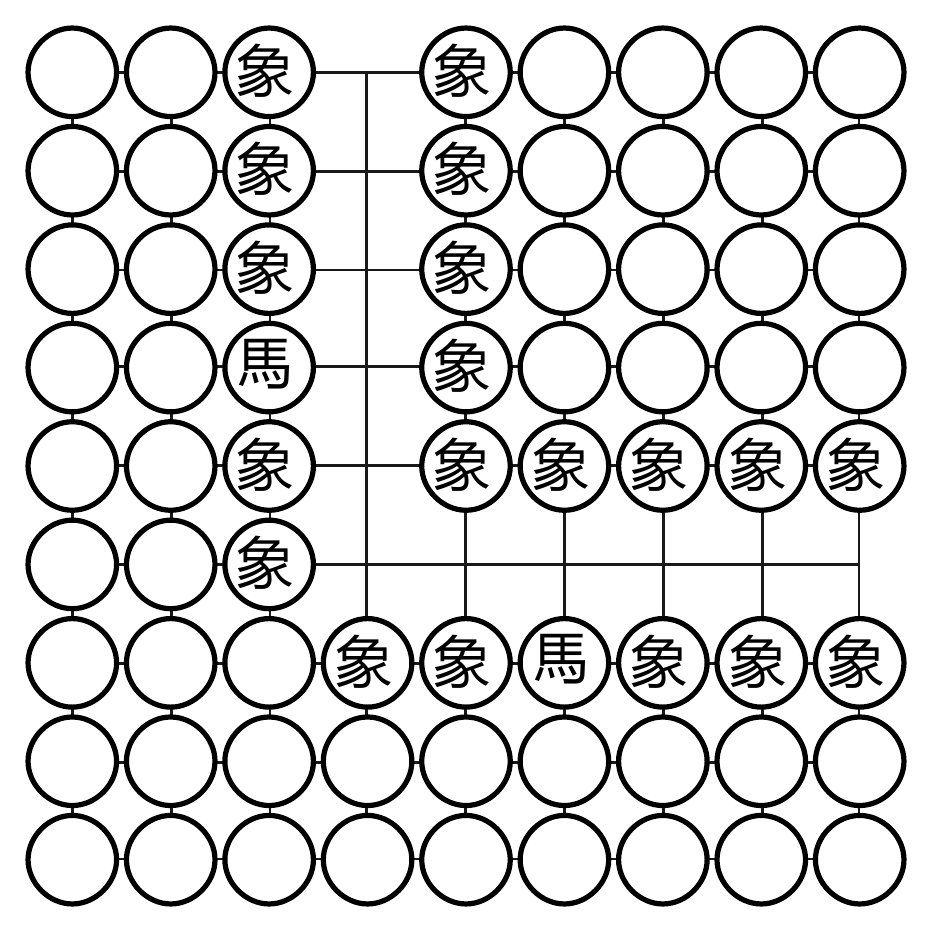}
		\caption{Turn gadget of Xiangqi.}
		\label{Xturn}
	\end{minipage}%
	\begin{minipage}[htbp]{0.5\linewidth}
		\centering
		\includegraphics[width=0.7 \linewidth]{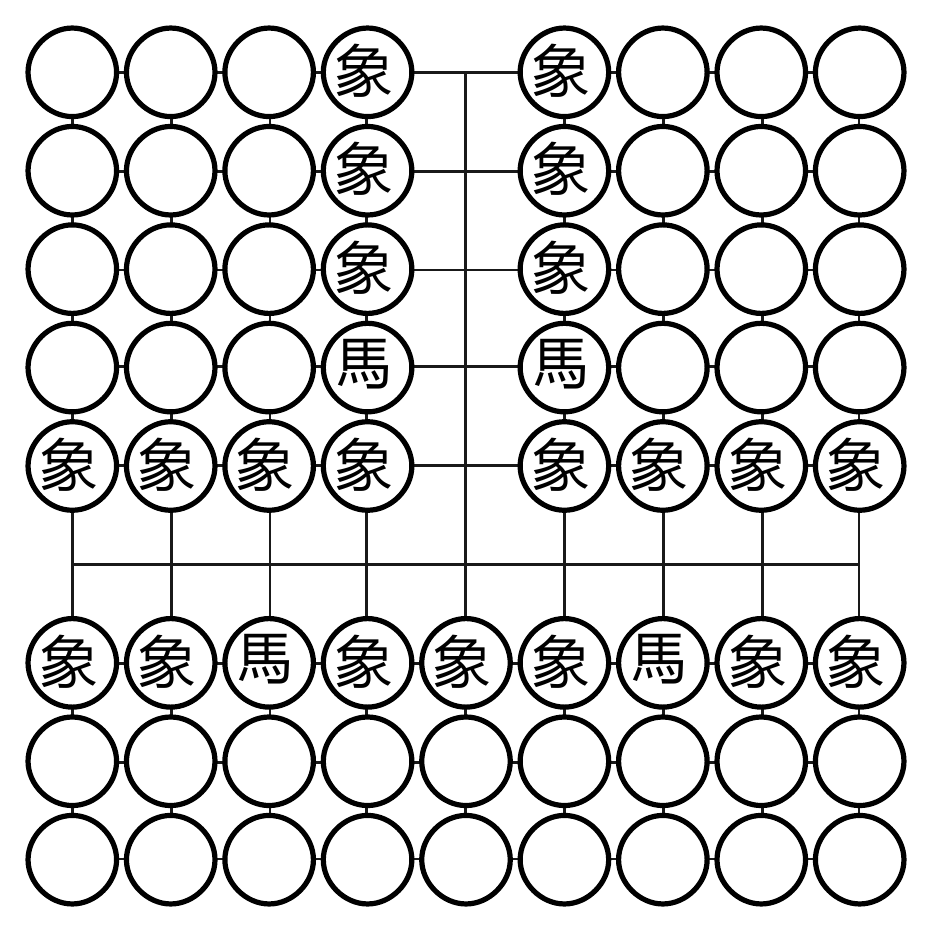}
		\caption{Switch and merge gadgets of Xiangqi.}
		\label{Xswitch}
	\end{minipage}
\end{figure}

\textbf{Finish gadget.} 
Figure \ref{Xfinish} illustrates a finish gadget of Xiangqi. 
Once main RC enters a finish gadget from south, the BG will be checkmated immediately. 
We assume that the BG in the finish gadget is at the corner of ``palace''.
Notice, in a finish gadget, the BG is moveable. 
Since ``stalemate" is a loss for the player with no legal moves in Xiangqi.

\textbf{Turn gadget.} 
Figure \ref{Xturn} illustrates a turn gadget of Xiangqi.  
Main RC can enter a turn gadget from north, and it can move east to leave. 

\textbf{Switch and merge gadgets.} 
Figure \ref{Xswitch} illustrates a switch gadget of Xiangqi. 
When main RC enters a switch gadget from north, it can move east or west to leave. 
Moreover, a merge gadget is identical to a switch gadget in Xiangqi.

\textbf{One-way gadget.} 
Figure \ref{Xoneway1} illustrates an one-way gadget of Xiangqi. 
Main RC can only traverse the one-way gadget from west to east.
If main RC attempts to traverse the gadget from east to west, the BR at s12 can move to s5 to obstruct main RC.

Here, we consider some nontrivial improper moves in an one-way door gadget:

(1) When main RC moves to e5, if the BR stops at f5 (or g5, h5,  ... , r5, t5, u5, v5) in order to obstruct main RC.
Main RC can capture the BH at c5, then the BR can not escape.

(2) When main RC moves to e5, if the BR stops at s5 in order to obstruct main RC.
Main RC can move to i5 to threaten the BH at i12. 
The BR should move back to s12 to protect the BH, otherwise main RC can capture the BH, and it can move to n12 safely, then it can capture three BH's (at n8, n9 and n13) continuously.
The variation diagram of the one-way gadget is illustrated in Figure \ref{Xoneway2}.
In this variation, two RR's at n15 and n16 become moveable, and they can capture the BR.
Furthermore, the RC at n14 becomes new main RC.

(3) If the RC at n14 captures the BE at n10, it will be captured by one of the BR's immediately.
Then a large number of BR's become moveable, thus two RR's (at n15 and n16) will be captured sooner or later.
So that main RC can not traverse the gadget any more.

(4) If the RR at n15 captures one of two BE's (at m15 and o15), it will be captured by one of two BC's (at k15 and q15) immediately.
It is no benefit to Red since red pieces are surrounded by powerful black pieces.

(5) If the RR at n16 captures one of BE's (at m16, n17 and o16), it will be captured by one of BR's so that a large number of BR's become moveable.

\begin{figure}[htbp]
	\centering  
	\includegraphics[width=0.87 \linewidth]{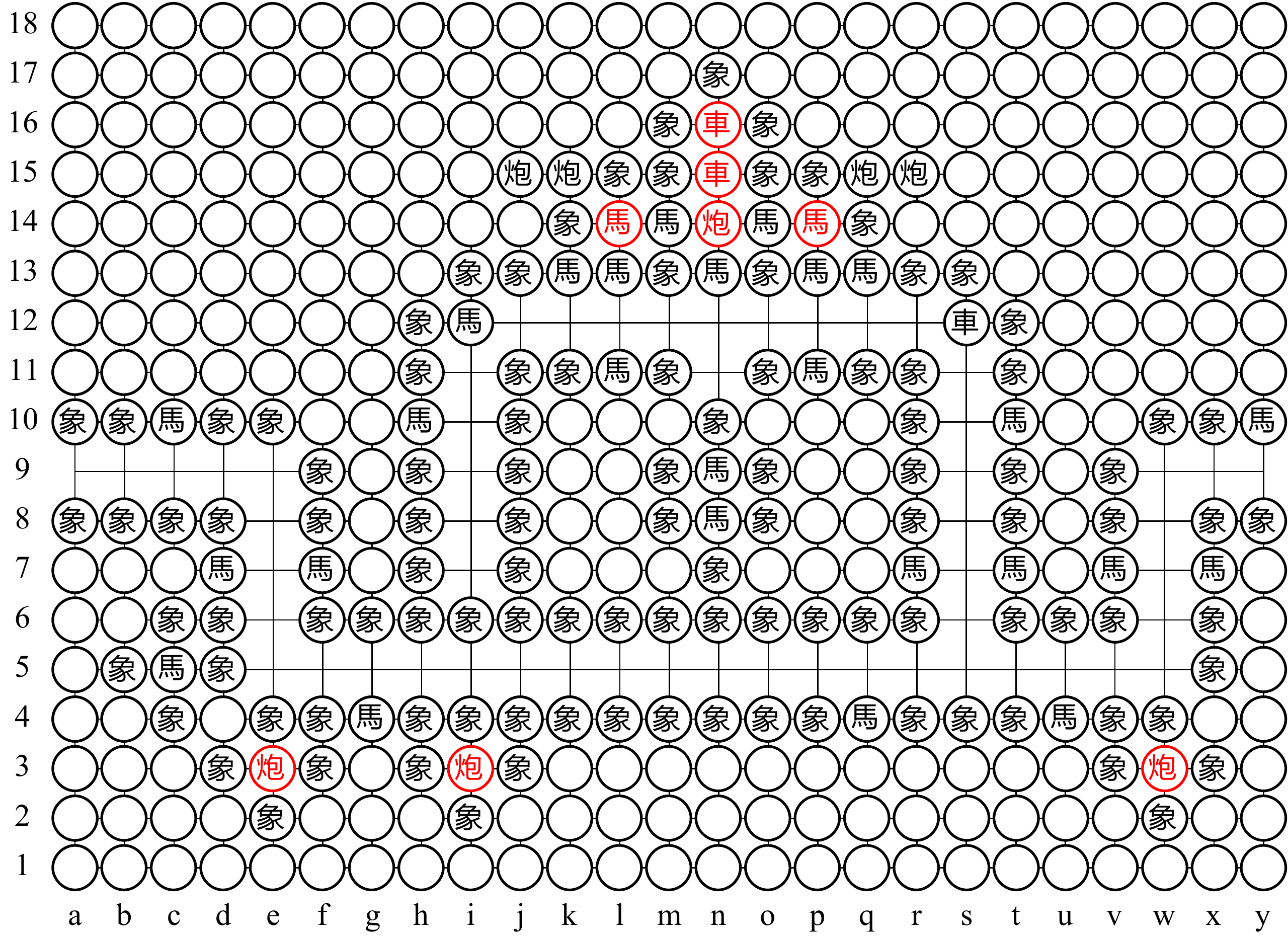}  
	\caption{One-way gadget of Xiangqi.}  
	\label{Xoneway1}   
\end{figure}

\begin{figure}[htbp]
	\centering  
	\includegraphics[width=0.87 \linewidth]{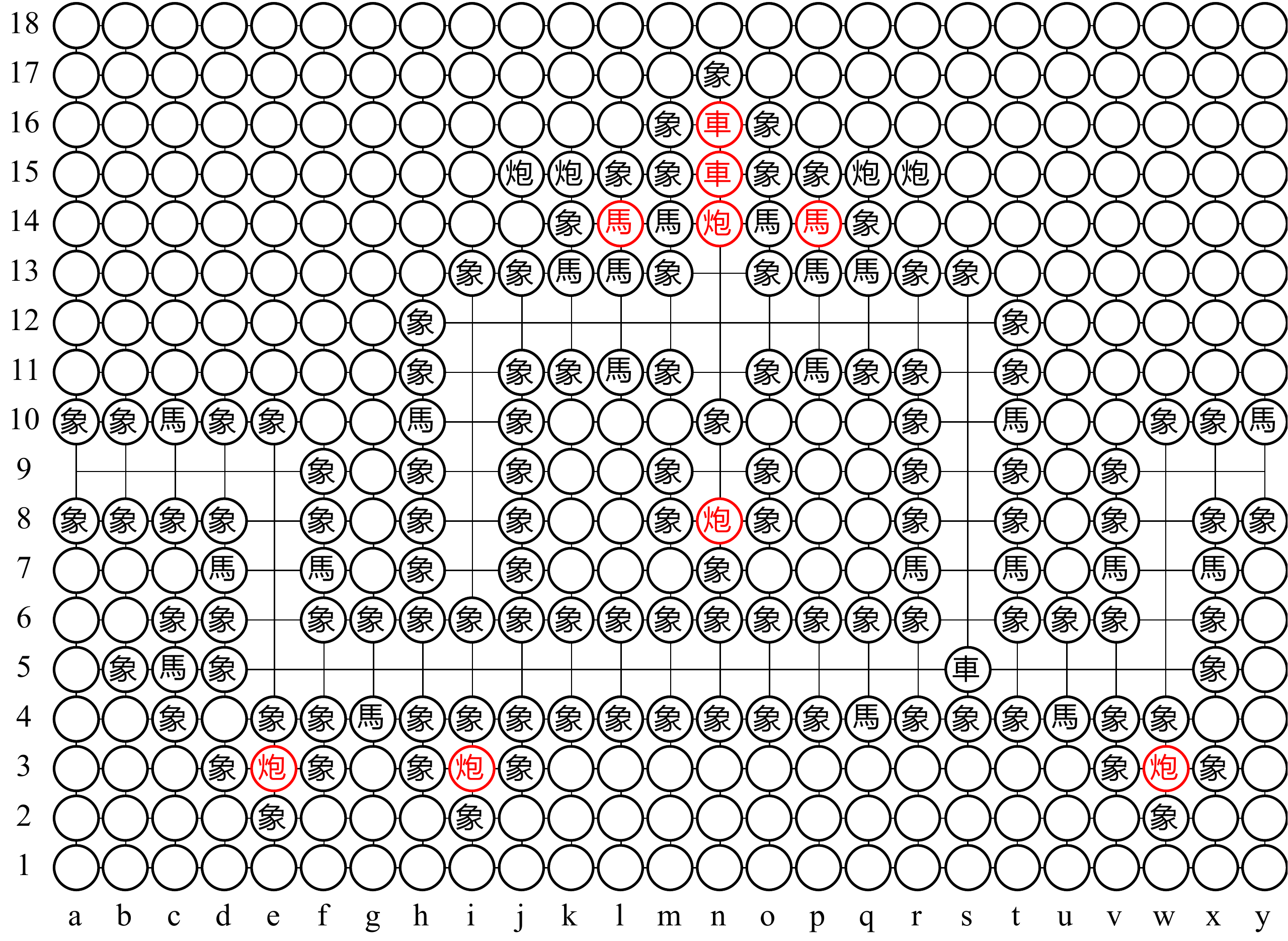}  
	\caption{Variation diagram of the one-way gadget.}  
	\label{Xoneway2}   
\end{figure}

\newpage

\textbf{Crossover gadget.} 
Figure \ref{Xcrossover1} illustrates a crossover gadget of Xiangqi. 
Main RC can traverse a crossover gadget through two paths, and there is no leakage between two paths. 
For one path, main RC enters the gadget from west, and it stops at d14, then the RC at r5 can move to t11 to force the BR to leave l11 so that main RC can traverse the gadget.
If the BR stops at d12 or d13, main RC can capture the BH at d16, so that the BR can not escape.
Similarly, if the BR stops at r12 (or r13, r14, ... , r17), main RC can move to r11, and it can capture the BH at r9.
Situation of another path is similar. 
Moreover, four RC's (at e7, f14, i16 and p18) prevent the BR from leaving the gadget.
The BR should move back to l11 when the RC moves back to r5, which prevents leakage between two paths.
The difference on number of moves of the RC and BR avoids repeating moves.

\begin{figure}[htbp]
	\centering  
	\includegraphics[width=0.83 \linewidth]{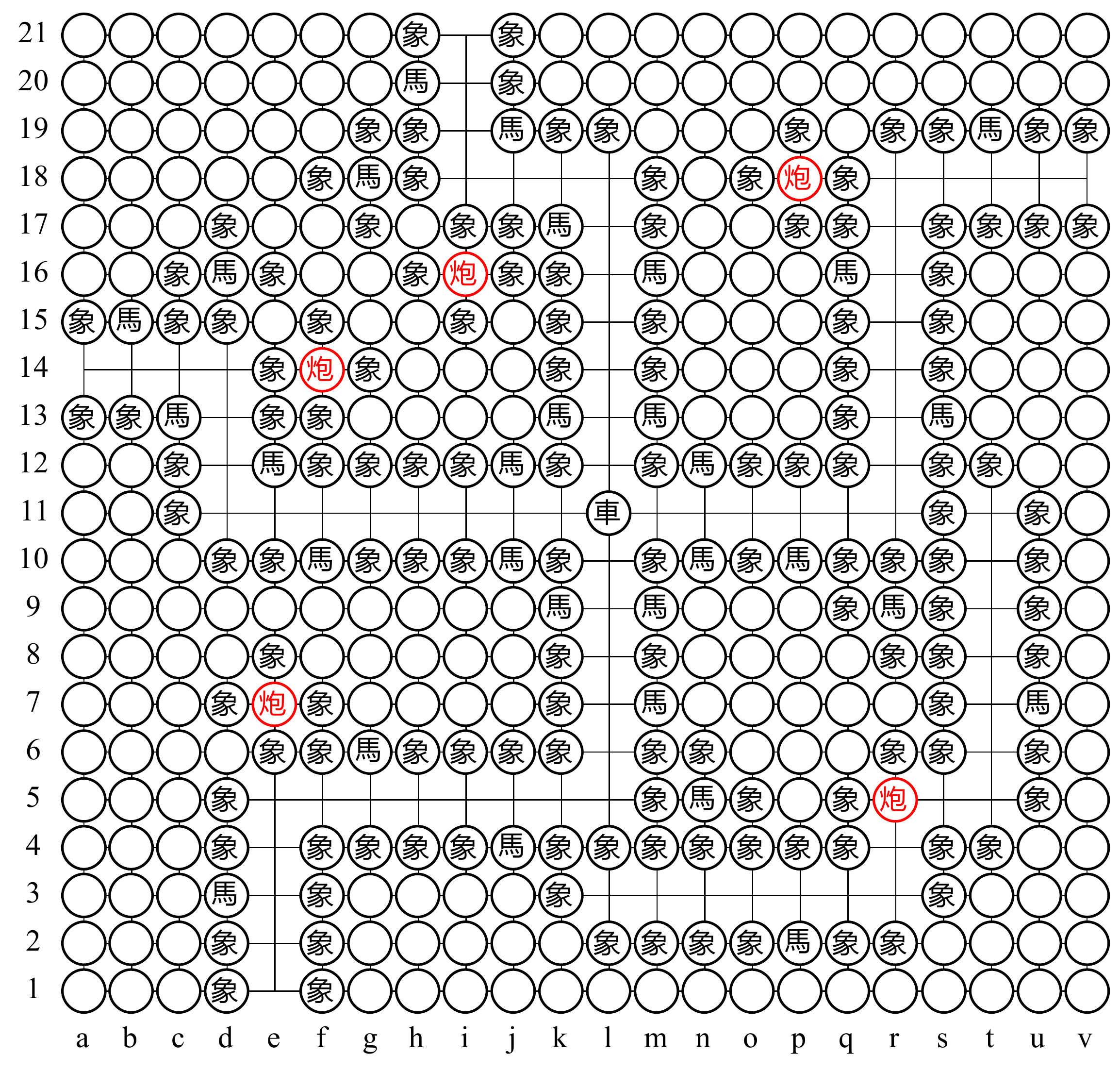}  
	\caption{Crossover gadget of Xiangqi.}  
	\label{Xcrossover1}   
\end{figure}

\textbf{Door gadget.} 
Figure \ref{Xdoor} illustrates a door gadget of Xiangqi. 
Points d8 and a5 are the entrance and the exit of the open path respectively. 
Path of point sequence (d8, d5, h5, c5, a5) is an open path.
Points k1 and n4 are the entrance and the exit of the traverse path respectively. 
Path of point sequence (k1, k4, n4) is a traverse path. 
In Figure \ref{Xdoor}, the gadget is in the closed state. 
Since the BC at i4 protects point k4, main RC can not traverse the traverse path. 
When main RC enters a door gadget from north, it can capture two BH's (at h5 and c5). 
Then the RH at h6 can capture the BC at i4, which makes the door gadget in the open state.

\begin{figure}[htbp]
	\centering  
	\includegraphics[width=0.55 \linewidth]{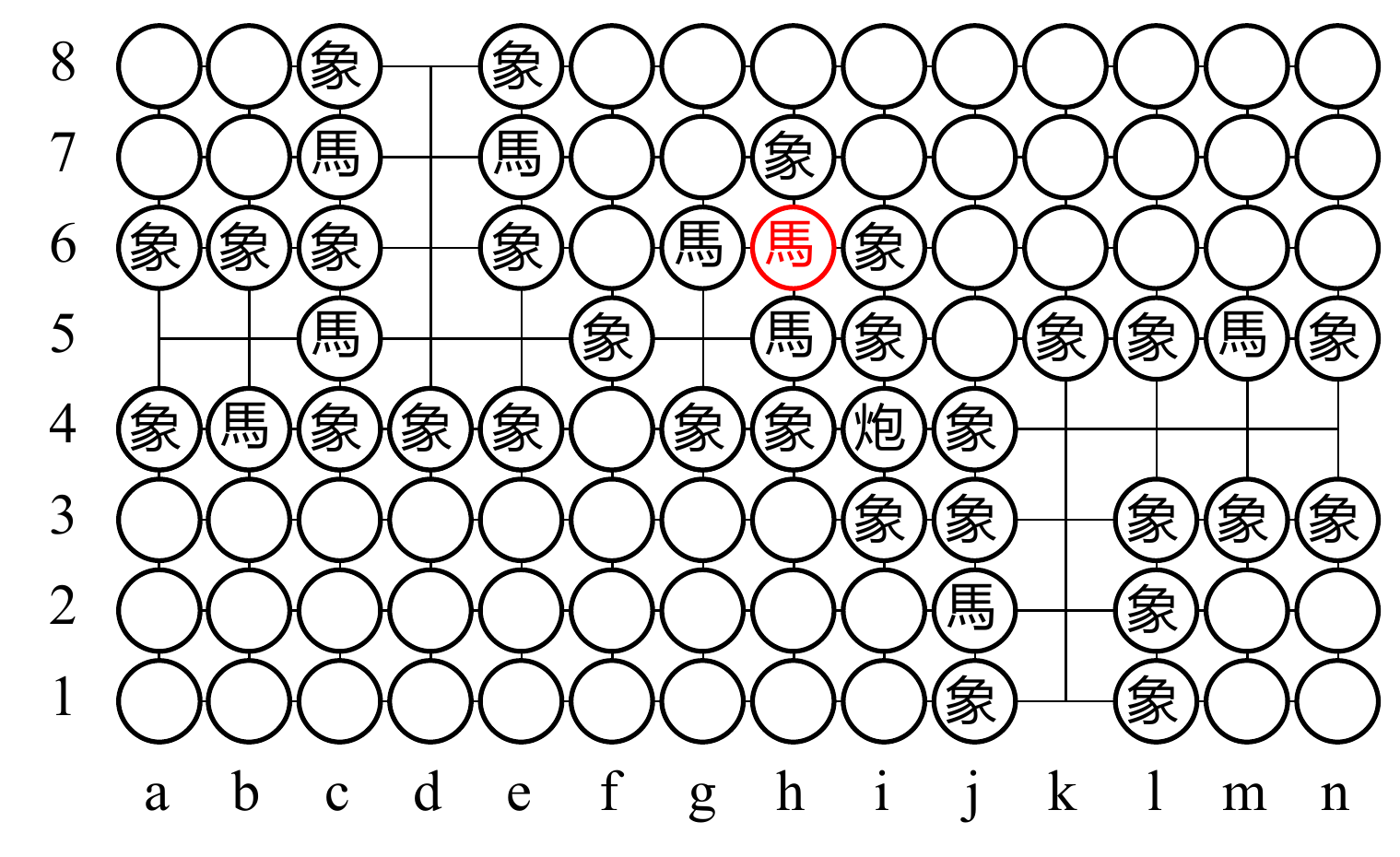}  
	\caption{Door gadget of Xiangqi.}  
	\label{Xdoor}   
\end{figure}

As all gadgets of NP-hardness framework have been constructed in Xiangqi, we obtain the following result.

\begin{proposition}
	Xiangqi is NP-hard.
\end{proposition}

\subsection{PSPACE-hardness of Xiangqi}

To prove PSPACE-hardness of Xiangqi, we need to construct all gadgets of PSPACE-hardness framework in Xiangqi.
Fortunately, the gadgets constructed in last section could be reused here.
So we just need to construct the open-close door gadget.

\textbf{Open-close (RRR) door gadget.} 
Figure \ref{XRRRdoor} illustrates an open-close door gadget of Xiangqi. 
In this gadget, only nine RC's (at q18, i3, i13, i27, i33, q3, q9, q23 and q33), one BC (at m18) and two BH's (at i17 and q19) can move. 
We call the RC at q18 and the BC at m18 control cannons.
When control RC stops at q18, the gadget is considered in the closed state; when control RC stops at i18, the gadget is considered in the open state.

\begin{figure}[htbp]
	\centering  
	\includegraphics[width=1.0 \linewidth]{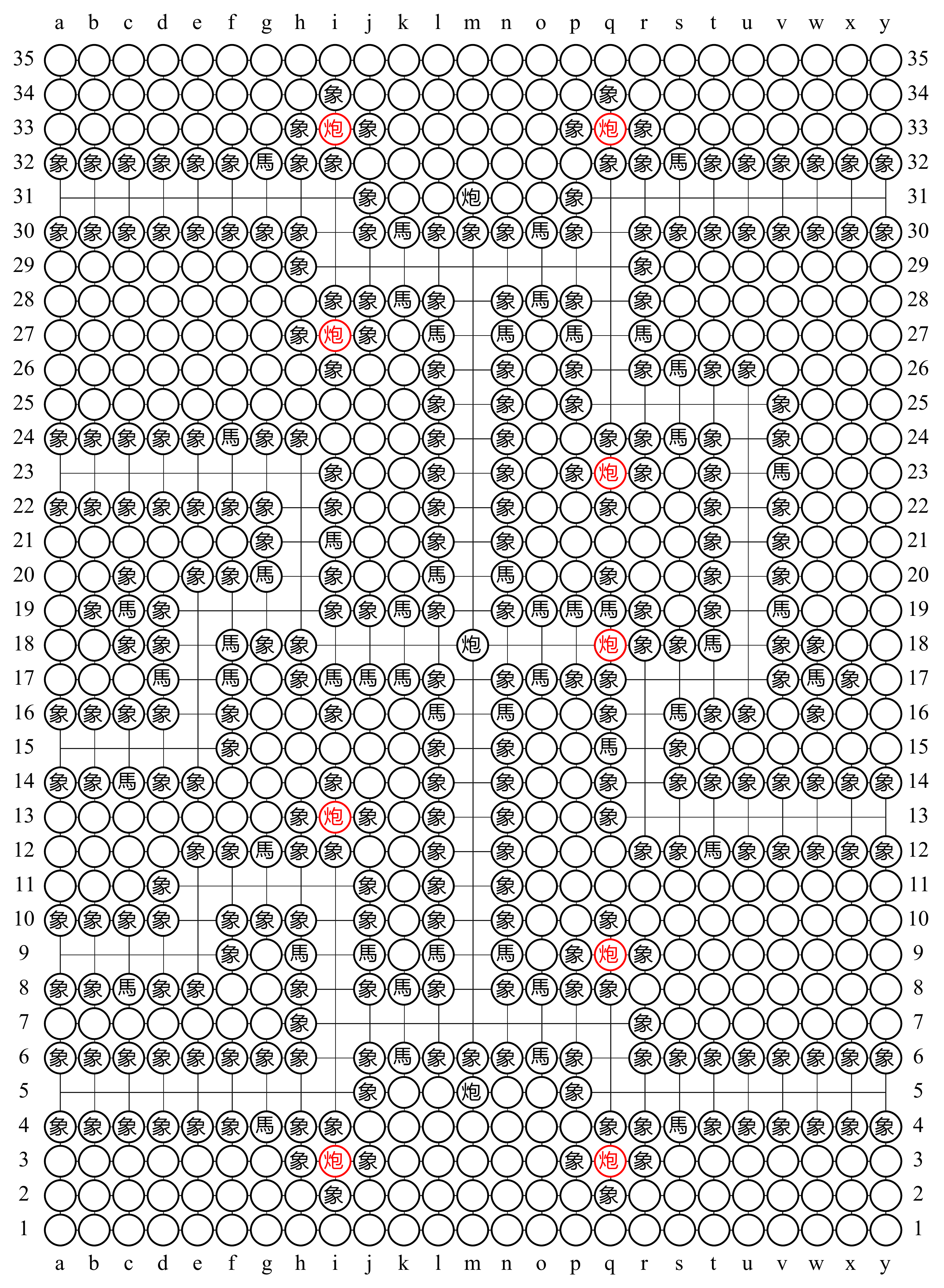}  
	\caption{Open-close (RRR) door gadget of Xiangqi.}  
	\label{XRRRdoor}   
\end{figure}

Points a5 and a9 are the entrance and the exit of the open path respectively. 
Path of point sequence (a5, i5, i7, i11, e11, e9, a9) is an open path. 
Points a15 and a23 are the entrance and the exit of the traverse path respectively. 
Path of point sequence (a15, e15, e19, h19, h23, a23) is a traverse path. 
Points y31 and y13 are the entrance and the exit of the close path respectively. 
Path of point sequence (y31, q31, q29, q25, u25, u17, r17, r13, y13) is a close path.
Points a31 and y5 are entrances of rapid checkmate paths.

Main RC can open the gadget by traversing the open path. 
When main RC moves to i7, control BC has to move to m7 to obstruct main RC. 
Otherwise, main RC can enter the rapid checkmate path. 
After control BC leaves m18, control RC can move to i18 to obstruct the BH at i17, which makes the gadget in the open state.
Then, when main RC moves to e11, control BC should move back to m18 to restrict movements of control RC.

Main RC can traverse the traverse path if and only if the gadget is in the open state.
If control RC stops at i18, main RC can traverse the path safely. 
If the gadget is in the closed state, main RC will be captured by the BH at i17 when it traverses the path.

When main RC traverses the close path, it has to close the gadget.
To traverse the path, main RC must pass point r17, however, the point is protected by the BE at q19. 
Thus, after main RC moves to q29, and it forces control BC to move to m29, control RC must move to q18 to obstruct the BH, which makes the gadget in the closed state.
Then, when main RC moves to u25, control BC should move back to m18 to restrict movements of control RC.

Here, we consider some nontrivial improper moves in an open-close gadget:

(1) If control RC attempts to leave the gadget, it must pass point m18. 
However, two BC's (at m5 and m31) protect point m18. 

(2) If control BC attempts to leave the gadget, it will be captured by one of eight RC's (at i3, i13, i27, i33, q3, q9, q23 and q33).

(3) If eight RC's move without capturing control BC.
Then a large number of BR's will become moveable, thus these moves do not benefit Red.

(4) If the BH at i17 moves to h19 without capturing main RC.
Main RC can capture the BH at c19 when it traverses the traverse path, and the BH at h19 can not escape.
If the BH at q19 moves to r17 without capturing main RC, main RC can capture the BH at w17 similarly.

(5) When main RC traverses the open path, and it moves to e11, control BC may move to m18. 
Then, if main RC moves back to i7, which forces control BC to move to m7 again. 
However, it is no benefit to Red, because Red requires two moves while Black just requires one move. 
The difference on number of moves also avoids repeating moves.

(6) When main RC traverses close path, if it moves back to q29. 
The situation is similar to case (5).

(7) If main RC moves into m column, it will be captured by one of the BC's (at m5 and m31).

(8) If control BC does not move back to m18 when main RC leaves the gadget.
Movements of control RC will not be restricted, then Red can open the gadget at any time.
Thus control BC should move back to m18 as far as possible.

As all gadgets of PSPACE-hardness framework have been constructed in Xiangqi, we obtain the following result.

\begin{proposition}
	Xiangqi is PSPACE-hard.
\end{proposition}

\subsection{EXPTIME-hardness of Xiangqi}

To prove EXPTIME-hardness of Xiangqi, we need to construct all gadgets of EXPTIME-hardness framework in Xiangqi.
The gadgets constructed in last two sections could be reused, and the start, finish, turn, switch and merge gadgets for Black could be constructed symmetrically.
We still need to construct the one-way gadget for Black and the crossover gadget for Red and Black.
Moreover, the RRR door gadget is identical to the open-close gadget in last section, and BBB door, BRB door and RRB door gadgets could be constructed symmetrically, so we also need to construct RBR door and BBR door gadgets.

Since paths are composed of red elephants and black elephants in our reduction, the border between unmoveable red rooks and black rooks must be taken into consideration, and it is illustrated in Figure \ref{Xborder}.
Notice, the elephants can not cross the river in Xiangqi.

\begin{figure}[htbp]
	\centering  
	\includegraphics[width=0.52 \linewidth]{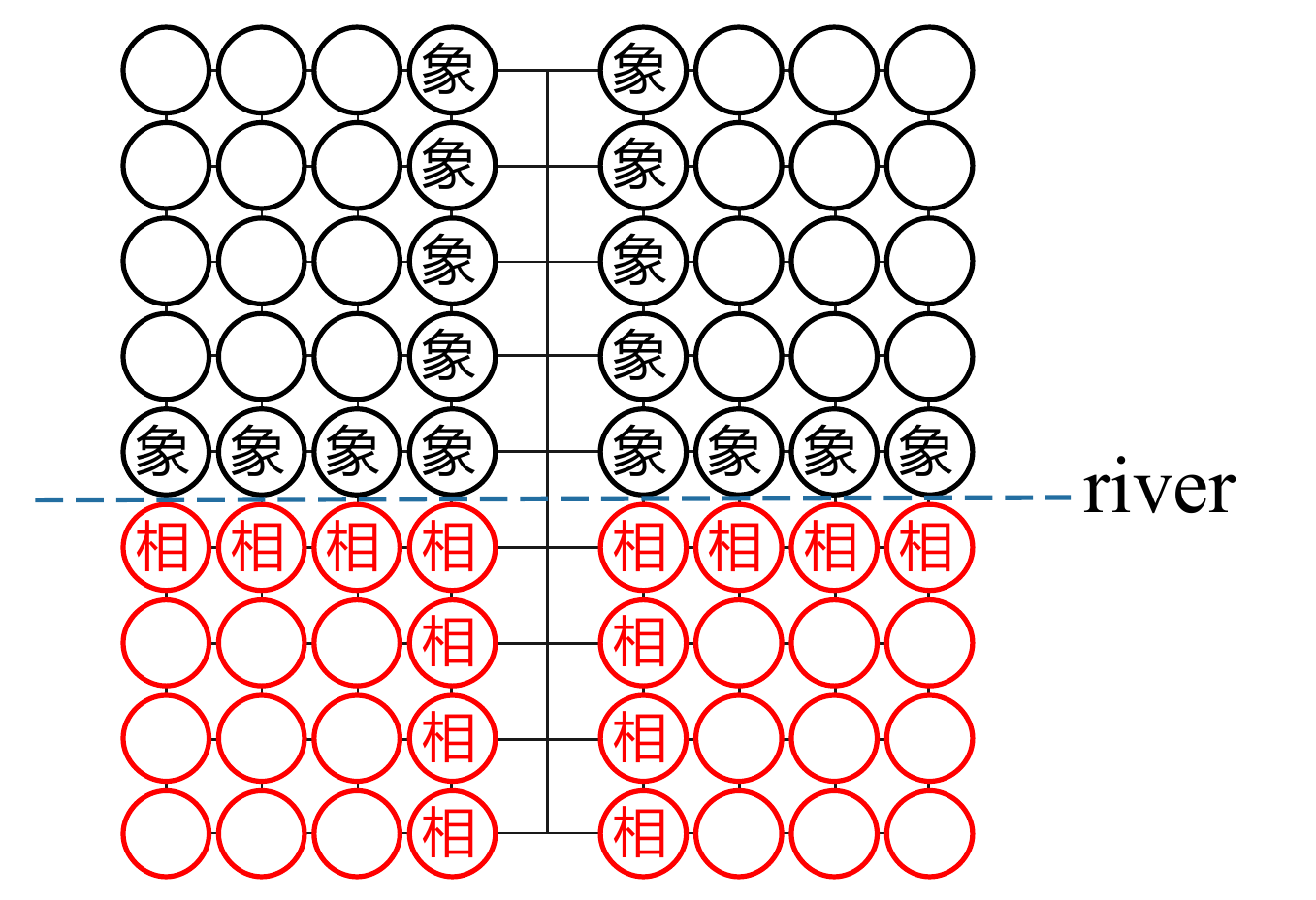}  
	\caption{Border between Red and Black in Xiangqi.}  
	\label{Xborder}   
\end{figure}

\begin{figure}[htbp]
	\centering  
	\includegraphics[width=0.75 \linewidth]{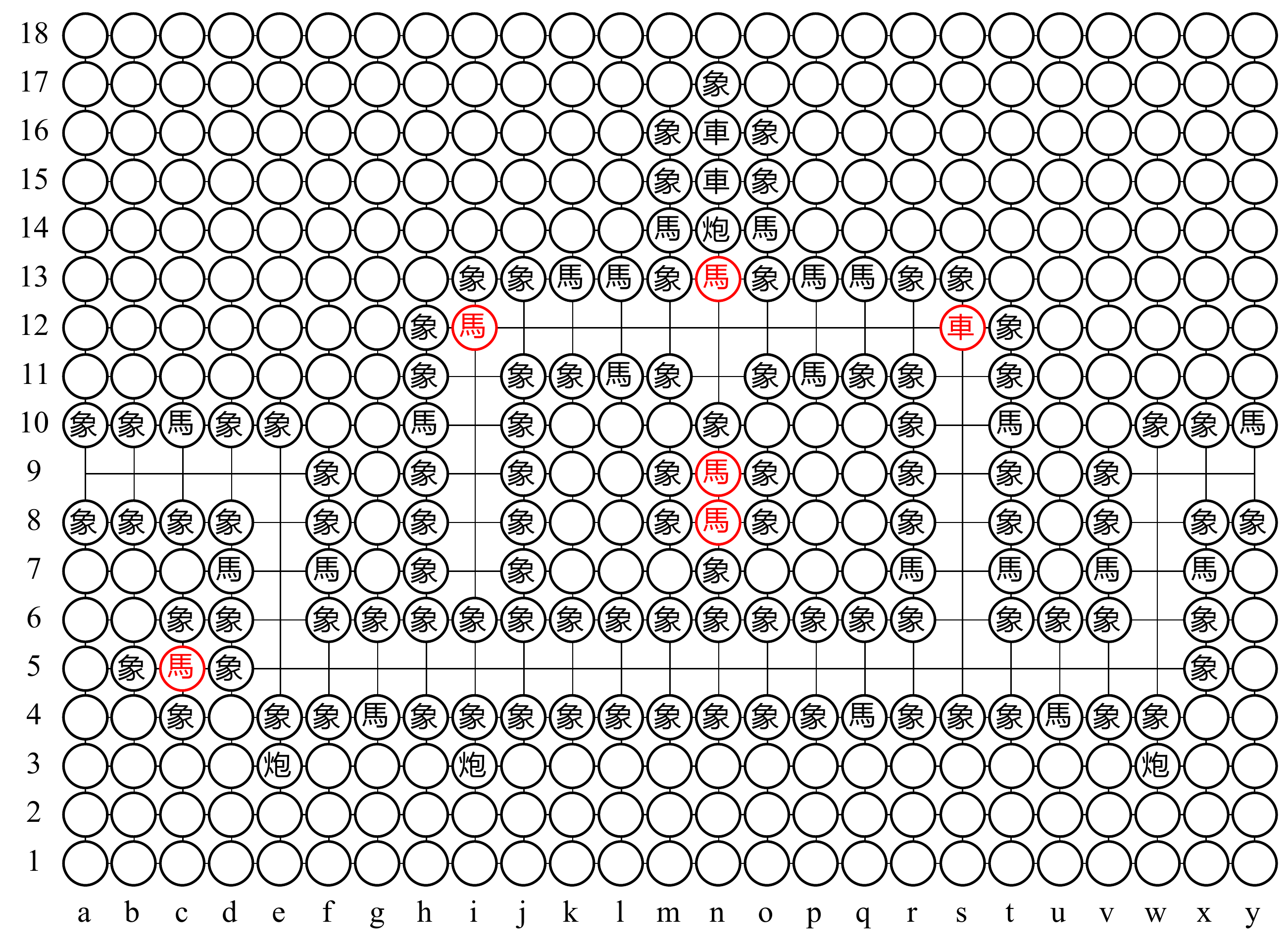}  
	\caption{One-way gadget for Black of Xiangqi.}  
	\label{Xoneway3}   
\end{figure}

\textbf{One-way gadget for Black.} 
Figure \ref{Xoneway3} illustrates an one-way gadget for Black of Xiangqi. 
This gadget is similar to the one-way gadget for Red (Figure \ref{Xoneway1}).
Main BC can only traverse the one-way gadget from west to east.
If main BC attempts to traverse the gadget from east to west, the RR at s12 can move to s5 to obstruct main BC.
If the RR stops at s5 to obstruct main BC when main BC moves from west to east, main BC can move to i5 to force the RR to move back to s12.
If the RR stops at s5, main BC can capture the RH at i12, and it can move to n12 so that three RH's (at n8, n9 and n13) will be captured.
Then one BC (at n14) and two BR's (at n15 and n16) become moveable.

\textbf{Crossover gadget for Red and Black.} 
Figure \ref{Xcrossover2} illustrates an crossover gadget (for Red and Black) of Xiangqi. 
Main RC can traverse the gadget from west to east, while main RC can traverse the gadget from north to south.
Two RC's (at c11 and o11) protect points c13 and o13, and two BC's (at l7 and l19) protect points j7 and j19, thus there is no leakage between two paths.

We consider a type of nontrivial improper moves in this gadget.
If main BC stops at d13 (or e13, f13, ... , n13) in order to obstruct main RC.
When main RC enters the gadget, and it moves to c13, main RC can ``jump'' over main BC, and it can capture two BH's (at s13 and s10) continuously.
Then the RC at s14 and the RR at s15 become moveable.
Moreover, the RC at s14 becomes new main RC, and the RR prevents main BC from traversing the gadget henceforth.
The situation that main RC stops at j8 (or j9, j10, ... , j18) is similar.

\begin{figure}[htbp]
	\centering  
	\includegraphics[width=0.8 \linewidth]{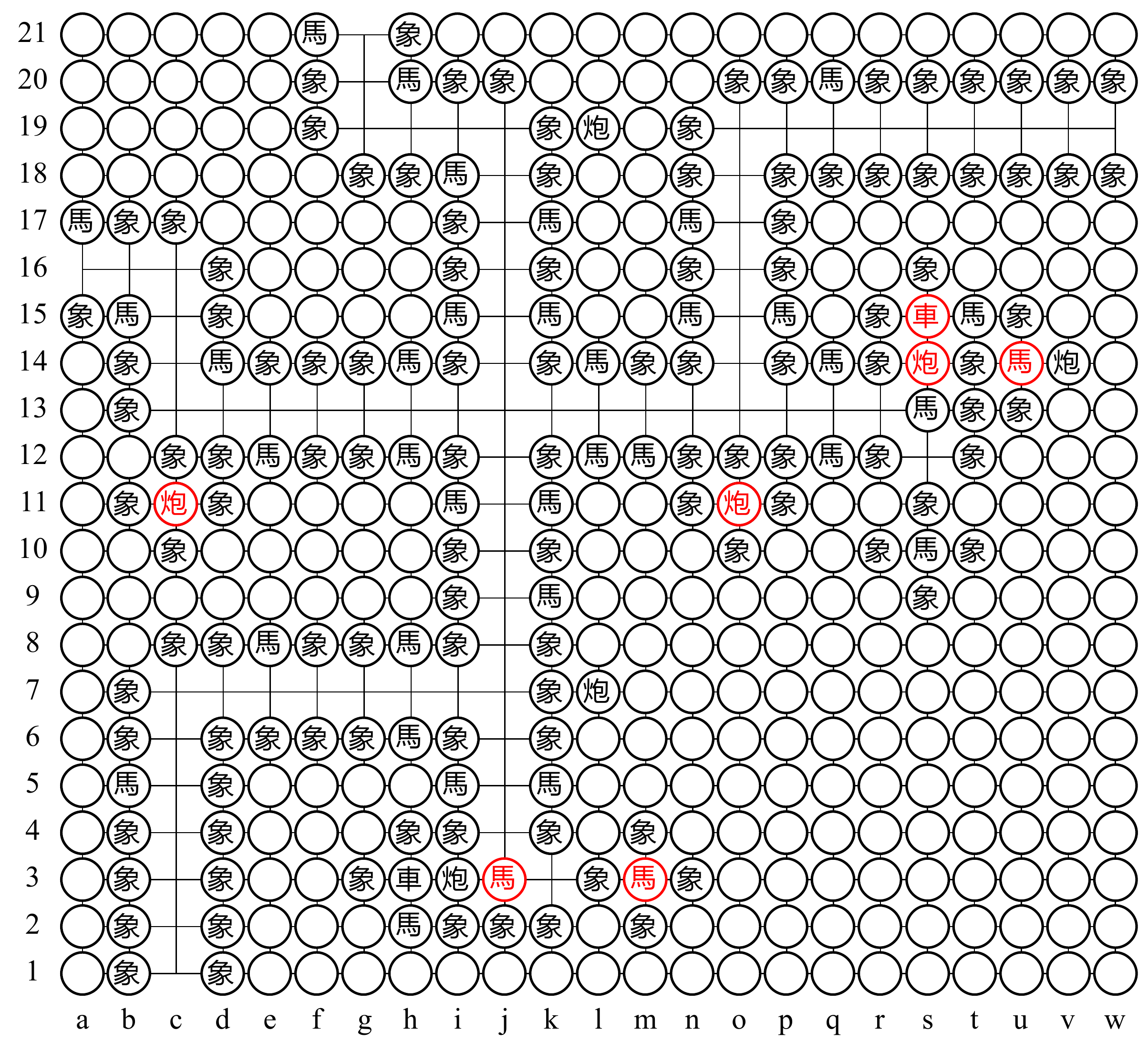}  
	\caption{Crossover gadget for Red and Black of Xiangqi.}  
	\label{Xcrossover2}   
\end{figure}

\textbf{RBR door gadget.} 
Figure \ref{XRBRdoor} illustrates a RBR door gadget of Xiangqi.  
In this gadget, only eleven RC's (at e18, c18, c20, i3, i13, i27, i33, q3, q9, q23 and q33), one BC (at m18) and one BH (at q19) can move. 
We call the moveable RC (at e18) and BC control cannons.
When control RC stops at e18, the gadget is considered in the closed state; when control RC stops at q18, the gadget is considered in the open state.

Points y31 and y13 are the entrance and the exit of the open path respectively. 
Path of point sequence (y31, q31, q29, q25, u25, u17, r17, r13, y13) is an open path for Red. 
Points a14 and a24 are the entrance and the exit of the traverse path respectively.
Path of point sequence (a14, e14, e24, a24) is a traverse path for Black. 
Points a5 and a9 are the entrance and the exit of the close path respectively.
Path of point sequence (a5, i5, i7, i11, e11, e9, a9) is a close path for Red.
Points a31 and y5 are entrances of rapid checkmate paths for Red.

When main RC traverses the open path, it has to open the gadget.
To traverse the path, main RC must pass point r17, however, the point is protected by the BH at q19. 
Thus, after main RC moves to q29, and it forces control BC to move to m29, control RC must move to q18 to obstruct the BH, which makes the gadget in the open state.
Then, when main RC moves to u25, control BC should move back to m18 to restrict movements of control RC.

Main BC can traverse the traverse path if and only if the gadget is in the open state.
If control RC stops at q18, main BC can traverse the path successfully. 
If the gadget is in the closed state, main BC can not traverse the path since it is obstructed by control RC.

Main RC can close the gadget by traversing the close path. 
When main RC moves to i7, control BC has to move to m7 to obstruct main RC. 
Otherwise, main RC can enter the rapid checkmate path. 
After control BC leaves m18, control RC can move to e18 to block the traverse path, which makes the gadget in the closed state.
Then, when main RC moves to e11, control BC should move back to m18 to restrict movements of control RC.

\begin{figure}[htbp]
	\centering  
	\includegraphics[width=1.0 \linewidth]{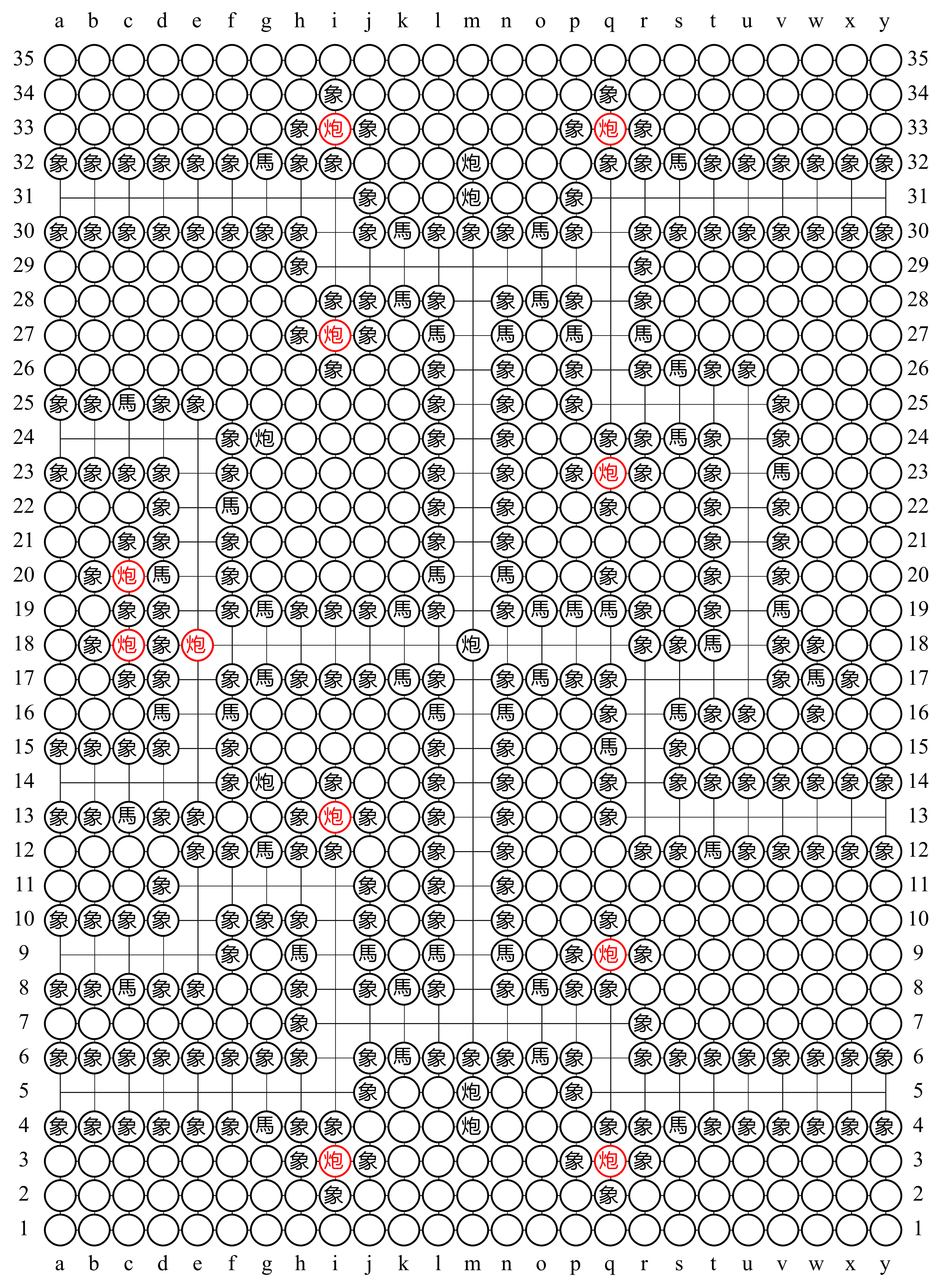}  
	\caption{RBR door gadget of Xiangqi.}  
	\label{XRBRdoor}   
\end{figure}

\begin{figure}[htbp]
	\centering  
	\includegraphics[width=1.0 \linewidth]{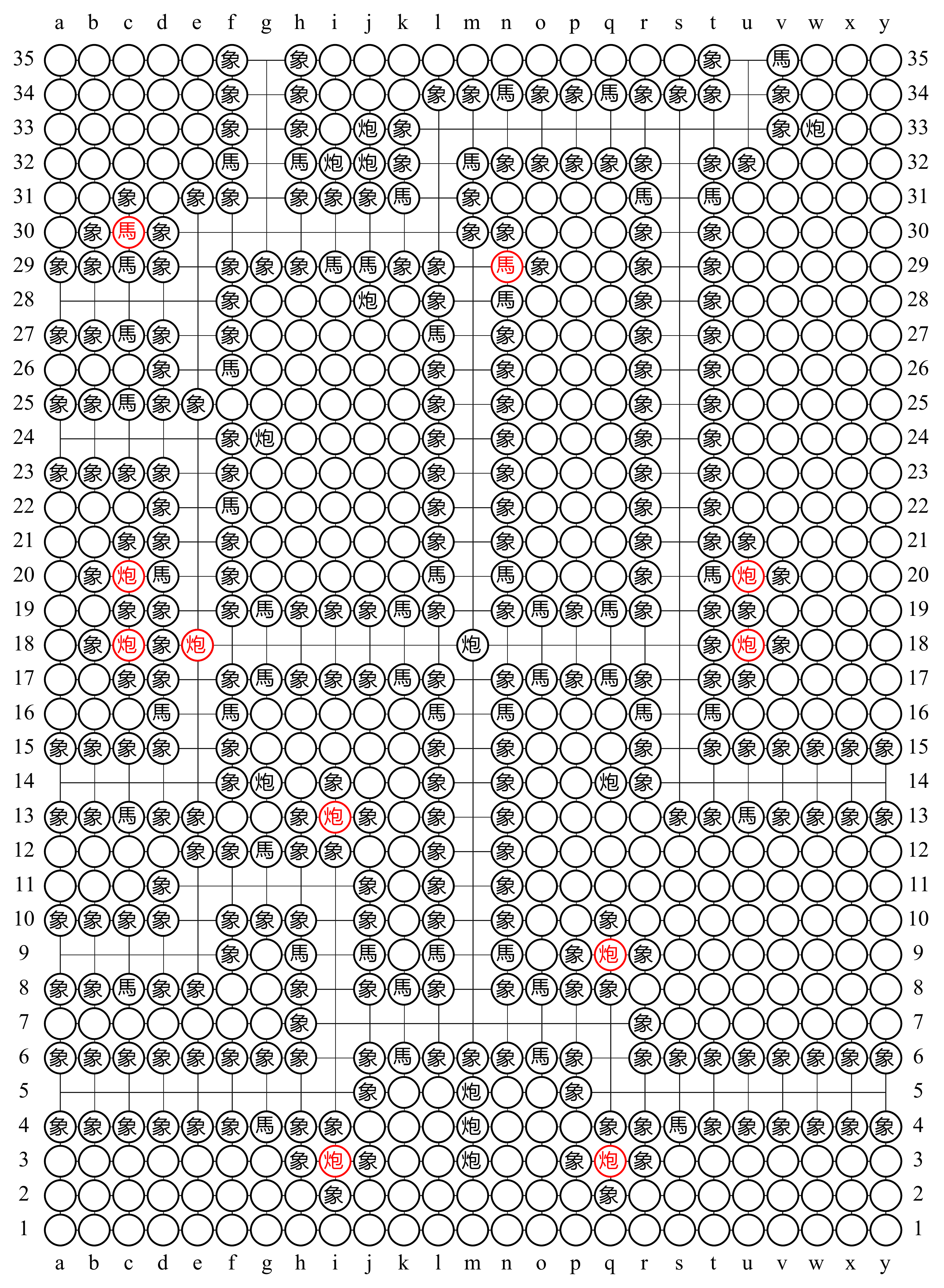}  
	\caption{BBR door gadget of Xiangqi.}  
	\label{XBBRdoor}   
\end{figure}

Here, we consider some different nontrivial improper moves in a RBR door gadget:

(1) Control RC can not leave the gadget.
Since two BC's (at g14 and g24) protect points e14 and e24.
Moreover, four BC's (at m4, m5, m31 and m32) prevent control RC from moving to m18 even if the RC at c18 can capture one of the BC's.

(2) When the gadget is in the open state, if control BC move to e18 in order to leave the gadget, it will be captured by the RC at c18.
If control BC captures the RC at c18, it will be captured by the RC at c20.
Then Red can close the gadget at any time.

(3) When the gadget is in the open state, if the RC at c18 captures control BC at m18, it will be captured by one of the BC's (at m5 and m31). 

(4) When the gadget is in the open state, if main BC moves to e18, it will be captured by the RC at c18.

\textbf{BBR door gadget.} 
Figure \ref{XBBRdoor} illustrates a BBR door gadget of Janggi. 
In this gadget, only nine RC's (at e18, c18, c20, i3, i13, q3, q9, u18 and u20), one RH (at n29) and one BC (at m18) can move. 
We call the moveable RC (at e18) and BC control cannons.
When control RC stops at e18, the gadget is considered in the closed state; when control RC stops at s18, the gadget is considered in the open state.

Points g35 and u35 are the entrance and the exit of the open path respectively. 
Path of point sequence (g35, g30, l30, l33, s33, u33, u35) is an open path for Black. 
Points a14 and a24 are the entrance and the exit of the traverse path respectively. 
Path of point sequence (a14, e14, e24, a24) is a traverse path for Black. 
Points a5 and a9 are the entrance and the exit of the close path respectively. 
Path of point sequence (a5, i5, i7, i11, e11, e9, a9) is a close path for Red.
Point y5 is the entrance of rapid checkmate path for Red, and points a28 and y14 are the entrances of rapid checkmate paths for Black.

Main BC can open the gadget by traversing the open path.
To traverse the path, main BC must pass point l30, however, the point is protected by the RH at n29. 
Thus control BC should move to m29 to obstruct the RH.
When main BC move to l33, control RC has to move to s18.
Otherwise, main BC can enter the rapid checkmate path.
Then control BC should move back to m18 to restrict movements of control RC.

The traverse path in a BBR door gadget is similar to the path in a RBR door gadget.
Main BC can traverse the traverse path if and only if the gadget is in the open state.

The close path in a BBR door gadget is also similar to the path in a RBR door gadget.
Main RC can close the gadget by traversing the close path.

Here, we consider some different nontrivial improper moves in a BBR door gadget:

(1) If control RC moves to s14, it will be captured by the BC at q14.

(2) If control RC moves to s33, it will be captured by one of the BC's (at j33 and w33).

(3) When control BC moves to m29 to obstruct the RH, if control RC moves to m18.
Three BC's (at m3, m4 and m5) prevent the RC's from leaving the gadget even if two RC's (at c18 and ) can capture two of three BC's.

(4) When control BC moves to m7 to obstruct main RC, if control RC moves to m18 to obstruct control BC.
In this case, three BC's can not force control RC to leave m18.
However, when main BC enters the open path, it can move to e30 so that control RC has to move to e18, otherwise, main BC can enter rapid checkmate path through a28.
Once control RC leaves m18, control BC can move back to m18.

(5) If control BC moves to s18, it will be captured by the RC at u18.

(6) If control BC captures the RC at u18, it will be captured by the RC at u20 immediately.

(7) If the RH (at n29) moves to l30 without capturing main BC, main BC can capture the RH at c30 to threaten the RH when it traverses the open path.
Then, there are five choices for the RH at l30.
If the RH stops at l30, it will be captured by main RC.
If the RH moves to j29 (or j31, k32, m32), it will be captured by one of three BC's (at i32, j28 and j32), then the moveable BC becomes new main BC.

(8) When main BC traverses the open path, it moves to l33, and control RC may move to s18. 
Then, if main BC moves back to l30, and it moves to e30 so that control RC has to move to e18 to prevent main BC from entering rapid checkmate path.
However, it is no benefit to Black, because Black requires two moves while Red just requires one move. 
The difference on number of moves also avoids repeating moves.

(9) If main RC stops at i7 for a long time, control BC has to stop at m7.
So that main BC can not traverse the open path.
The situation is similar to case (6) of BBR door gadget of Janggi.

(10) If main BC stops at s33 for a long time, the main RC can not close the BBR door gadget.
The situation is similar to case (7) of BBR door gadget of Janggi.

As all gadgets of EXPTIME-hardness framework have been constructed in Xiangqi, we obtain the following result.

\begin{theorem}
	Xiangqi is EXPTIME-complete.
\end{theorem}

\begin{proof}
	Firstly, Xiangqi could be solved by a brute search algorithm in exponential time, thus Xiangqi is in EXPTIME. 
	Secondly, we have constructed all gadgets of EXPTIME-hardness framework in Xiangqi, thus Xiangqi is EXPTIME-hard.
\end{proof}

\textbf{Remarks.}
In our reduction, the game board is almost full of pieces.
If we assume that the whole instance of Xiangqi is surrounded by unmoveable elephants, advisors and soldiers, number of pieces could be reduced.
We do not discuss the reachability of the instance of Xiangqi here.

\newpage

\section{Conclusion}

We review NP-hardness framework and PSPACE-hardness framework for 2D platform games. 
By using similar technique, we introduce a EXPTIME-hardness framework for proving hardness of games. 
With these frameworks in hand, we can analyse computational complexity of games conveniently. 
We discuss complexity of Xiangqi and Janggi, and we prove that Xiangqi and Janggi are both EXPTIME-complete.

Establishing more hardness frameworks is one of future research directions.
Robson \cite{expspacegame} described several EXPSPACE-complete formula games which are ``no-repeat" version of games described in \cite{difficultgames}.
A further rule modification results in a formula game which is 2EXPTIME-complete.
However, it is difficult to simulate the ``no-repeat" rule in the framework with polynomial gadgets.

We will also try to use the hardness frameworks to analyse complexity of other interesting games in future.


\bibliographystyle{plain}
\bibliography{ref}

\end{document}